\documentclass[11pt,tightenlines,nofootinbib,notitlepage,superscriptaddress]{revtex4-1}

\usepackage[linesnumbered,ruled,vlined]{algorithm2e}
\usepackage{graphicx}
\usepackage{multirow}
\usepackage{amsmath}
\usepackage{amssymb}
\usepackage{dsfont}
\usepackage{amsthm}
\usepackage{eucal}
\usepackage{bm}
\usepackage{upgreek}
\usepackage{mathrsfs}
\usepackage{color}
\usepackage{latexsym}
\usepackage{xcolor}
\usepackage{float}
\usepackage{dsfont}
\usepackage{enumerate}
\usepackage{stackengine}
\usepackage{braket}
\usepackage[titletoc]{appendix}

\theoremstyle{plain}
\newtheorem{thm}{Theorem}
\newtheorem{proposition}[thm]{Proposition}

\newtheorem{corollary}[thm]{Corollary}
\newtheorem{lemma}[thm]{Lemma}
\newtheorem{definition}[thm]{Definition}

\theoremstyle{remark}
\newtheorem{rem}{Remark}
\newtheorem*{note}{Note}

\usepackage{hyperref}
\hypersetup{colorlinks,linkcolor={blue},citecolor={red},urlcolor={olive}}

\newcommand{\bb}{\begin{equation}}
\newcommand{\ee}{\end{equation}}
\newcommand{\bbb}{\begin{equation*}}
\newcommand{\eee}{\end{equation*}}

\newcommand{\ketbra}[2]{\ket{#1}\!\!\bra{#2}}

\newcommand{\hr}[1]{\hat{\rho}_{#1}}
\newcommand{\bs}[1]{\boldsymbol{#1}}

\newcommand*{\coloneqq}{\mathrel{\vcenter{\baselineskip0.5ex \lineskiplimit0pt \hbox{\scriptsize.}\hbox{\scriptsize.}}} =}
\newcommand{\tclass}[1]{\mathcal{T}(\mathcal{H}_{#1})}
\newcommand{\id}{\mathds{1}}

\newcommand{\texteq}[1]{\overset{\text{#1}}{=}}
\newcommand{\textleq}[1]{\overset{\text{#1}}{\leq}}
\newcommand{\textgeq}[1]{\overset{\text{#1}}{\geq}}

\stackMath
\newcommand\xxrightarrow[2][]{\mathrel{%
  \setbox2=\hbox{\stackon{\scriptstyle#1}{\scriptstyle#2}}%
  \stackunder[0pt]{%
    \xrightarrow{\makebox[\dimexpr\wd2\relax]{$\scriptstyle#2$}}%
  }{%
   \scriptstyle#1\,%
  }%
}}

\newcommand{\tends}[2]{\xxrightarrow[\! #2 \!]{\mathrm{#1}}}
\newcommand{\tendsk}[1]{\xxrightarrow[\! k\rightarrow \infty\!]{\mathrm{#1}}}

\DeclareMathOperator{\Tr}{Tr}
\begin{document}
\title{All phase-space linear bosonic channels are approximately Gaussian dilatable}

\author{Ludovico Lami}
\affiliation{School of Mathematical Sciences and Centre for the Mathematics and Theoretical Physics of Quantum Non-Equilibrium Systems, University of Nottingham, University Park, Nottingham NG7 2RD, United Kingdom}
\email{ludovico.lami@gmail.com}

\author{Krishna Kumar Sabapathy}
\affiliation{Quantum Information and Communication, \'{E}cole polytechnique de Bruxelles, CP 165, Universit\'{e} libre de Bruxelles, 1050 Bruxelles, Belgium}
\affiliation{Xanadu, 372 Richmond St W, Toronto ON, M5V 2L7, Canada}
\email{krishnakumar.sabapathy@gmail.com}

\author{Andreas Winter}
\affiliation{F\'{\i}sica Te\`{o}rica: Informaci\'{o} i Fen\`{o}mens Qu\`{a}ntics, Departament de F\'{i}sica, Universitat Aut\`{o}noma de Barcelona, ES-08193 Bellaterra (Barcelona), Spain}
\affiliation{ICREA -- Instituci\'o Catalana de Recerca i Estudis Avan\c{c}ats, Pg.~Lluis Companys 23, ES-08010 Barcelona, Spain}
\email{andreas.winter@uab.cat}

\begin{abstract}
We compare two sets of multimode quantum channels acting on a finite collection of harmonic oscillators: (a) the set of linear bosonic channels, whose action is described as a linear transformation at the phase space level; and (b) Gaussian dilatable channels, that admit a Stinespring dilation involving a Gaussian unitary.
Our main result is that the set  (a) coincides with the closure of (b) with respect to the strong operator topology. We also present an example of a channel in (a) which is not in (b), implying that taking the closure is in general necessary. This provides a complete resolution to the conjecture posed in Ref.\ [K.K.~Sabapathy and A.~Winter,~\href{https://journals.aps.org/pra/abstract/10.1103/PhysRevA.95.062309}{Phys.\ Rev.\ A 95, 062309 (2017)}]. Our proof technique is constructive, and yields an explicit procedure to approximate a given linear bosonic channel by means of Gaussian dilations. It turns out that all linear bosonic channels can be approximated by a Gaussian dilation using an ancilla with the same number of modes as the system. We also provide an alternative dilation where the unitary is fixed in the approximating procedure. Our results apply to a wide range of physically relevant channels, including all Gaussian channels such as amplifiers, attenuators, phase conjugators, and also non-Gaussian channels such as additive noise channels and photon-added Gaussian channels.  The method also provides a clear demarcation of the role of Gaussian and non-Gaussian resources in the context of linear bosonic channels. 
Finally, we also obtain independent proofs of classical results such as the quantum Bochner theorem, and  develop some tools to deal with convergence of sequences of quantum channels on continuous variable systems that may be of independent interest. 
\end{abstract}

\maketitle 

\tableofcontents

\section{Introduction} \label{sec intro}

In this work we investigate fundamental properties of quantum transformations on modes of electromagnetic radiation field, an example of a continuous variable system, i.e.\ systems whose degrees of freedom are continuous in nature~\cite{loock-review}. In particular, we study optical implementations of what are known as linear bosonic channels, as introduced by Holevo and Werner~\cite{hw01}. The term `linear' has many connotations in quantum optics, but in this article linear channels are those for which the input signal undergoes a linear transformation when described at the level of phase space characteristic functions, i.e.\ phase-space linear (see Eq.~\eqref{linear chi} below). Linearity at the level of density operators is always assumed. 

Although the class of linear bosonic channels is special in many respects, it turns out to encompass many examples of physically relevant channels. For instance, all bosonic Gaussian channels~\cite{raul-rmp} are linear, as well as some non-Gaussian operations such as general additive classical noise channels~\cite{idler}, photon-added Gaussian channels~\cite{pang}, to list a few.  Linear channels are instrumental in obtaining benchmarks for teleportation and storage of squeezed states~\cite{wolfbench}, and have been investigated from an information-theoretic point of view with respect to reversibility~\cite{Shirokov2013} and extremality~\cite{Holevo2013}. Capacities of the general additive noise channels have been studied in both the classical~\cite{verdu} and quantum settings~\cite{stefan}.

An important observation is that the set of linear channels contains all so-called `Gaussian dilatable channels', defined as those that can be obtained through quadratic interactions between the system and an ancillary environment prepared in an arbitrary state~\cite{pang}. We will explore this connection in detail in this article. Note that Gaussian dilatable channels can be implemented in a relatively easy way when compared to general non-Gaussian transformations, while still retaining some of the interesting features of the latter. This is particularly important as it is known that non-Gaussian resources, notwithstanding the complexity of their harnessing, are indispensable for many quantum information processing and quantum computation protocols~\cite{andersen2015hybrid}.

To motivate the need to go beyond the Gaussian formalism in quantum optics, consider e.g.\ that in spite of the rich structure Gaussian entanglement exhibits~\cite{simon00, Werner-Wolf, Giedke2001, revisited, LL-log-det}, it turns out that it can not be distilled using Gaussian operations alone~\cite{ent-gauss1, ent-gauss2, ent-gauss3}. More generally, a similar no-go result holds for generic state conversion tasks in arbitrary Gaussian resource theories~\cite{G-resource-theories}. 
On the other hand, non-Gaussian operations are provably necessary to realize universal quantum computation~\cite{seth} and many other quantum information processing tasks~\cite{cerf2005non, adesso2009optimal,ohliger2010limitations, sabapathy2011robustness, andersen2013high, sabapathy2018states, dell2007continuous, su2018implementing,walschaers2018tailoring}.
In view of these limitations, it will not come as a surprise that a substantial effort has been put into developing a consistent resource theory of non-Gaussianity. Many non-Gaussianity measures have been proposed and studied in the past decade~\cite{Genoni2008, Genoni2010, ivan2012measure,Marian2013, Takagi2018,albarelli2018resource}, that can be applied e.g.\ to bound the conversion rates between arbitrary states by means of Gaussian operations~\cite{Takagi2018}. Recently, a resource theory of non-Gaussianity for channels has also been put forth~\cite{quntao18}.

Another reason to study Gaussian dilatable channels besides the fact that they constitute one of the few analytically treatable classes of non-Gaussian operations, is that they provide a systematic way of investigating classes of operations with certain physically meaningful properties, e.g.\ those that can be implemented by means of passive optics and arbitrary states~\cite{idler} or passive optics with passive ancillary states~\cite{Jabbour2018}, the latter being motivated from a thermodynamic context, and also for obtaining the operator-sum representations of the corresponding channels \cite{ivan2011operator,pang}.

In view of their operational and theoretical importance, in this paper we study the set of linear bosonic channels in great detail. Our main result establishes that every linear bosonic channel can be approximated by a sequence of Gaussian dilatable channels.
Mathematically speaking, we prove that \emph{the closure of the set of Gaussian dilatable channels in the strong operator topology coincides with the set of linear bosonic channels} (Theorem~\ref{approx n modes thm}). Taking the closure is necessary, as we show that there are linear bosonic channels that have no Gaussian dilations, even when one allows the ancillary state to have infinite energy. An example is provided by the `binary displacement channel', which displaces the input state by $+s$ or $-s$ (with $s$ fixed) with equal probabilities (Corollary~\ref{Es not G dilatable cor}). Our results solve the question posed in~\cite[Conjecture~1]{pang} (see also~\cite[Remark~5]{Shirokov2013}).

Remarkably, our solution is entirely explicit, and in fact it gives also a feasible experimental procedure to approximate the action of any desired linear bosonic channel by means of a single (possibly non-Gaussian) state and Gaussian unitary dynamics. If this unitary is allowed to vary with the degree of approximation, the ancillary system can be chosen to have the same number of modes as the system on which the channel acts (Corollary~\ref{approx n modes cor}). Incidentally, this also entails that although $2n$ modes may be required for an exact dilation of an $n$-mode Gaussian channel~\cite{min-dilations}, only $n$ ancillary modes suffice if we choose to approximate instead (Remark~\ref{approx dilations G channels rem}). In order to accommodate possible experimental feasibility of our approximation procedure, we also consider the case when the Gaussian unitary is necessarily fixed, and one can only vary the ancillary state. In this setting we are able to construct strong operator approximations of any given $n$-mode linear bosonic channel that require an ancillary system with $n+k$ modes, where $k$ is a number that depends only on the matrix that implements the (linear) phase space transformation induced by the channel (Proposition~\ref{approx n+k modes prop}).

Our exposition is meant to be entirely self-contained. Along the way, we give independent and simplified proofs of many classical results such as the quantum Bochner theorem (Lemma~\ref{quantum Bochner}) and the complete positivity condition for linear bosonic channels (Lemma~\ref{lemma cp linear channels}), which may be of independent interest. Our aim is to guide the reader through some of the subtleties of convergence in infinite dimension, and to provide a small handbook of handy convergence results that are of broad applicability in quantum optics. Most notably, we dig out of previous literature a very handy lemma to establish convergence of sequences of density operators in various topologies (the `SWOT convergence' Lemma~\ref{SWOT convergence lemma}), and we use it to give an analogous criterion for the convergence of sequences of bosonic channels (the `SWOTTED convergence' Lemma~\ref{SWOTTED convergence channels}). We demonstrate how this tool can be used to verify convergence in the strong operator topology -- equivalently, uniformly on energy-bounded states -- almost effortlessly by applying it to the family of Gaussian additive noise channels that model the transformations induced by the Braunstein-Kimble~\cite{BK-teleportation} continuous variable teleportation protocol (Remark~\ref{BK convergence rem}).

The rest of the paper is structured as follows. In \S~\ref{sec bosonic channels} we introduce the basic formalism (\S~\ref{subsec phase space}), define linear bosonic channels (\S~\ref{subsec linear channels}), and discuss operator topologies (\S~\ref{subsec operator topologies}). The following \S~\ref{sec main} is devoted to the presentation of our main results. In \S~\ref{subsec linear channels closed} we start by showing that linear bosonic channels form a strong operator closed set, and in \S~\ref{subsec approximate} we construct Gaussian dilatable approximations for all such channels. In \S~\ref{subsec closure necessary} we discuss an example of a linear bosonic channel that is not exactly Gaussian dilatable, while in \S~\ref{subsec number auxiliary} we provide a more experimentally feasible procedure to implement the above approximations. Finally, in \S~\ref{sec conclusions} we summarise our contributions and highlight some open problems.


\section{Bosonic states and channels} \label{sec bosonic channels}

\subsection{Phase-space formalism} \label{subsec phase space}

Let us consider a system of $n$ electromagnetic modes described as quantum harmonic oscillators. The associated Hilbert space is the space of square integrable functions in $2n$ real variables, denoted by $\mathcal{H}_n\coloneqq L^2(\mathbb{R}^{2n})$. The quadrature operators $x_j,p_k$ ($j,k=1,\ldots, n$) can be conveniently grouped together to form the vector $r \coloneqq (x_1,p_1,x_2,p_2,\cdots,x_n,p_n)^\intercal$. The canonical commutation relations then take the form 
\bb
[r_j,r_k] = i\Omega_{jk}\, ,
\label{CCR}
\ee
where
\bb
\Omega \coloneqq \begin{pmatrix} 0 & 1 \\ -1 & 0 \end{pmatrix}^{\oplus n}
\label{Omega}
\ee
is the standard symplectic form. In what follows we denote by $\mathcal{T}(\mathcal{H}_n)$ the space of trace-class operators acting on $\mathcal{H}_n$.

A particularly important role is played by the family of unitary Weyl--Heisenberg displacement operators, defined as
\bb
D(\xi) \coloneqq e^{i \xi^\intercal \Omega r}\, ,
\label{D}
\ee
for all $\xi\in \mathbb{R}^{2n}$. Applying the Baker--Campbell--Hausdorff formula and making use of Eq.~\eqref{CCR} we see that
\bb
D(\xi_{1}) D(\xi_{2}) = e^{-\frac{i}{2} \xi_{1}^\intercal\Omega \xi_{2}} D(\xi_{1}+\xi_{2})\, ,
\label{Weyl}
\ee
referred to as the Weyl form of the canonical commutation relations. In fact, the formulation in Eq.~\eqref{Weyl} is preferable to that in Eq.~\eqref{CCR} in many respects, not least because it involves bounded (unitary) operators instead of unbounded ones.

Among the important properties of these operators, we recall that the associated coherent states $\ket{\lambda}\coloneqq D(\lambda)\ket{0}$ satisfy the completeness relation
\bb
\braket{\alpha|\beta} = \int \frac{d^{2n}\lambda}{(2\pi)^{n}}\, \braket{\alpha|\lambda}\braket{\lambda|\beta}\qquad \forall\ \ket{\alpha},\, \ket{\beta}\in \mathcal{H}_n\, ,
\label{completeness explicit}
\ee
which we can also symbolically write as
\bb
\int \frac{d^{2n}\lambda}{(2\pi)^{n}}\, \ket{\lambda}\!\!\bra{\lambda} = I\, ,
\label{completeness}
\ee
where $I$ is the identity operator. At this level, we regard Eq.~\eqref{completeness} as a purely formal representation of Eq.~\eqref{completeness explicit}, so that we do not need to worry about the convergence of the above integral in the operator sense. Mathematically, one says that the integral in Eq.~\eqref{completeness} is understood to converge in the weak operator topology. For a brief introduction to operator topologies, we refer the reader to \S~\ref{subsec operator topologies}.

Quantum states are described by density operators, i.e.\ positive semidefinite trace-class\footnote{A positive operator $A$ acting on a Hilbert space is said to be trace-class if it has finite trace, i.e.\ if \unexpanded{$\sum_{m} \braket{m|A|m} <\infty$} for some (and thus all) orthonormal bases \unexpanded{$\{\ket{m}\}_m$}.} operators of unit trace. The set of all density operators on a Hilbert space $\mathcal{H}$ will be denoted by $\mathcal{D}(\mathcal{H})$. To every trace-class operator $T\in \mathcal{T}(\mathcal{H}_n)$ we can associate a \textbf{characteristic function} $\chi_T:\mathbb{R}^{2n}\rightarrow \mathbb{C}$, defined as 
\bb
\chi_T(\xi)\coloneqq \Tr \left[ T D(\xi)\right] .
\label{chi}
\ee
Characteristic functions are important because they encode all the information about their parent operator, which can be reconstructed as~\cite[Corollary 5.3.5]{HOLEVO}
\bb
T = \int \frac{d^n\xi}{(2\pi)^n}\, \chi_T(\xi) D(-\xi)\, ,
\label{integral}
\ee
where again the integral converges in the weak operator topology. As originally proved in~\cite{Pool1966} (see also~\cite[\S~5.3]{HOLEVO}), the above correspondence $T\leftrightarrow \chi_{T}$, which we have defined only for trace-class $T$, can in fact be extended to an isometry between the Hilbert space of Hilbert--Schmidt operators on $\mathcal{H}_n$ and that of square-integrable functions on $\mathbb{R}^{2n}$, denoted by $L^2(\mathbb{R}^{2n})$. The fact that the mapping is an isometry can be expressed through the noncommutative Parseval's identity~\cite[Eq.~(5.3.22)]{HOLEVO}
\bb
\Tr \left[ T_{1}^{\dag} T_{2} \right] = \int \frac{d^{2n}\xi}{(2\pi)^{n}} \, \chi_{T_{1}}(\xi)^{*} \chi_{T_{2}}(\xi)\, .
\label{Parseval}
\ee

We can ask what are the conditions a given function $f:\mathbb{R}^{2n}\rightarrow\mathbb{C}$ must satisfy to ensure that it is the characteristic function of a quantum state. To answer this question we introduce some terminology.

\begin{definition} \emph{\cite{Kastler65}.} \label{A positivity def}
Given a skew-symmetric $2n\times 2n$ matrix $A$, a function $f:\mathbb{R}^{2n}\rightarrow \mathbb{C}$ is said to be $A$-positive if
for all finite collections of vectors $\xi_1,\ldots, \xi_N\in \mathbb{R}^{2n}$ one has 
\bb
    \left( f(\xi_\mu - \xi_\nu) e^{\frac{i}{2}\xi_\mu^\intercal A \xi_\nu} \right)_{\mu,\nu=1,\ldots, N}\geq 0\, ,
    \label{A positivity}
\ee
meaning that the matrix on the l.h.s.\ is positive semidefinite.
\end{definition}

\begin{rem} \label{inversion A positive rem}
It is elementary to observe that any $A$-positive function $f:\mathbb{R}^{2n}\rightarrow \mathbb{C}$ (for $A$ skew-symmetric) must satisfy $f(-\xi) = f(\xi)^*$ for all $\xi\in \mathbb{R}^{2n}$, and in particular $f(0)$ must be real (and nonnegative). In fact, it can be seen that an $A$-positive function $f$ remains such under inversion of the argument, i.e.\ the new function $g$ defined by $g(\xi)\coloneqq f(-\xi)$ is again $A$-positive.
\end{rem}

The answer to the above question regarding the physical validity of a characteristic function is then given in terms of the following `quantum Bochner theorem', established in~\cite{Kastler65, Loupias66} (see also~\cite[Theorem~5.4.1]{HOLEVO}). In Appendix~\ref{app quantum Bochner} we provide a direct proof which is independent of the analogous result for classical probability theory and does not seem to have appeared in the literature before.

\begin{lemma}[Quantum Bochner Theorem] \emph{\cite{Kastler65,Loupias66}.} \label{quantum Bochner}
A complex-valued function $f:\mathbb{R}^{2n}\rightarrow \mathbb{C}$ on $\mathbb{R}^{2n}$ is the characteristic function of a density operator if and only if the following three conditions are satisfied:
\begin{enumerate}[(i)]
    \item $f(0)=1$;
    \item $f$ is continuous at $0$; and
    \item $f$ is $\Omega$-positive in the sense of Definition~\ref{A positivity def}, where $\Omega$ is given by Eq.~\eqref{Omega}.
\end{enumerate}
\end{lemma}

\begin{rem}
It is known that every function that meets requirements (i), (ii) and (iii) of the above Lemma~\ref{quantum Bochner} is necessarily bounded in modulus by $1$ and continuous everywhere (see Lemma~\ref{continuity Omega positive lemma}).
\end{rem}


We conclude this brief exposition by introducing quantum \textbf{Gaussian states}. These can be equivalently defined as thermal states of Hamiltonians that are quadratic in the canonical operators $r_j$, or as those density operators whose characteristic function is a Gaussian. If
\bb
\chi_\rho(\xi) = e^{ - \frac14 \xi^\intercal \Omega^\intercal V\Omega\xi + is^\intercal \Omega \xi}
\label{chi Gaussian}
\ee
for some $2n\times 2n$ real matrix $V$ and some vector $s\in \mathbb{R}^{2n}$, we say that $\rho$ is a Gaussian state with \textbf{covariance matrix} $V$ and \textbf{mean} $s$ (compare with~\cite[Eq.~(4.48)]{BUCCO}). Clearly, Gaussian states are uniquely specified by these two quantities. It can be seen that the function on the r.h.s.\ of Eq.~\eqref{chi Gaussian} is the characteristic function of a quantum state if and only if
\bb
V + i \Omega \geq 0 \, ,
\label{Heisenberg}
\ee
which can be regarded as a manifestation of Heisenberg's uncertainty principle~\cite[Eq.~(3.77)]{BUCCO}.  This well-known fact can also be deduced from Lemma~\ref{A positivity lemma} in Appendix~\ref{app quantum Bochner}. The vacuum state $\ketbra{0}{0}$ is an example of a Gaussian state; its mean is $0$ and its covariance matrix is $V=\id$, hence
\bb
\chi_{\ketbra{0}{0}}(\xi) = e^{- \frac14 \xi^\intercal \xi} .
\label{chi vacuum}
\ee

\subsection{Linear bosonic channels} \label{subsec linear channels}

One of the fundamental notions of quantum information theory is that of a quantum channel. A quantum channel describes transformations induced on open physical systems by unitary interactions with an environment that is subsequently discarded~\cite[Chapter~8]{NC}. In our context, we can define a \textbf{bosonic channel} acting on $n$ modes as a completely positive and trace-preserving linear map $\Phi$ acting on the set of trace-class operators $\tclass{n}$. It is sometimes convenient to describe the action of the channel at the level of characteristic functions, by writing $\Phi:\chi_{\rho}(\xi)\mapsto \chi_{\Phi(\rho)}(\xi)$ and giving an explicit expression for this latter function. Equivalently, we can also switch to the Heisenberg picture and specify instead the action of $\Phi^\dag$ on all displacement operators.

Particularly simple examples of bosonic channels are the so-called \textbf{linear bosonic channels}, that act as
\bb
\Phi_{X,f}^\dag \left( D(\xi) \right) \coloneqq D(X\xi) f(\xi)\, ,
\label{linear dag}
\ee
where $X$ is a $2n\times 2n$ real matrix, and $f:\mathbb{R}^{2n}\rightarrow\mathbb{C}$ is a complex-valued function. We can rewrite this transformation at the level of characteristic functions as
\bb
\chi_\rho(\xi) \mapsto \chi_{\Phi(\rho)}(\xi) = \chi_\rho(X\xi) f(\xi)\, .
\label{linear chi}
\ee
Because of the simplicity of their phase space action, linear bosonic channels have consistently played a major role in theoretical quantum optics~\cite{hw01, Holevo2013, Shirokov2013}. It is natural to ask what conditions on $X$ and $f$ ensure that $\Phi_{X,f}$ is a legitimate quantum channel. The first solution of this problem was put forward in~\cite{Demoen77} (see also~\cite{Evans1977, Demoen79}). We provide a self-contained proof in Appendix~\ref{app cp linear channels}.

\begin{lemma} \emph{\cite{Demoen77}.} \label{lemma cp linear channels}
A map $\Phi_{X,f}$ whose action is given by~\eqref{linear chi} is completely positive and trace-preserving, and hence a linear bosonic channel, if and only if:
\begin{enumerate}[(i)]
    \item $f(0)=1$;
    \item $f$ is continuous at $0$; and
    \item $f$ is $J(X)$-positive according to Definition~\ref{A positivity def}, where
\bb
J(X) \coloneqq \Omega - X^\intercal\Omega X\, .
\label{J(X)}
\ee
\end{enumerate}
\end{lemma}


Observe that every Gaussian channel~\cite[\S~5.3]{BUCCO} is a linear bosonic channel, but the converse fails to hold. The simplest example of linear bosonic channel that is not a Gaussian operation is probably the \textbf{additive noise channel}, defined by
\begin{align}
\rho &\longmapsto \int D(-s)\rho D(s)\, \mu\left(d^{2n}s\right) , \label{additive noise p} \\
\chi_\rho(\xi) &\longmapsto \chi_\rho(\xi) f(\xi)\, . \label{additive noise chi}
\end{align}
where $\mu$ is an arbitrary probability measure over $\mathbb{R}^{2n}$, and $f$ is its Fourier transform
\bb
f(\xi) \coloneqq \int e^{i\xi^\intercal \Omega s} \mu\left(d^{2n} s\right) .
\label{Fouries p and f}
\ee

In a way, linear bosonic channels can thus be thought of as a natural generalisation of Gaussian channels. While this point of view may be mathematically well motivated, it is not operationally satisfactory, because Eq.~\eqref{linear chi} does not tell us anything about how to implement a given linear bosonic channel in a physically feasible way. To amend this we can decide to look at \textbf{Gaussian dilatable channels} instead~\cite{pang}. By definition, a Gaussian dilatable channel acts as
\bb
\rho \longmapsto \Tr_E \left[ U_{AE} (\rho_A \otimes \sigma_E) U_{AE}^\dag \right] ,
\label{G dilatable Stinespring}
\ee
where: $E$ is an $m$-mode optical system; $U_{AE}$ is a Gaussian unitary on the bipartite system $AE$ (obtained by combining arbitrary displacements on $A$ and $E$ with symplectic unitaries on $AE$~\cite[\S~5.1.2]{BUCCO}); and $\sigma_E$ is an \emph{arbitrary} state of the system $E$. For a pictorial representation of Eq.~\eqref{G dilatable Stinespring}, see Figure~\ref{fig3}. 
From a practical point of view, remember that the Gaussian unitary $U_{AE}$ can be implemented by means of multimode interferometers (passive optics) and single-mode squeezers~\cite{dutta94}. In turn, multimode interferometer can be decomposed into two-mode beam splitters and phase-shifters~\cite{triangle,square}, so that these two operations together with single-mode squeezers suffice to reproduce the action of $U_{AE}$.
If we want to specify that an $m$-mode ancillary state suffices to implement a Gaussian dilation, we say that the channel is \textbf{Gaussian dilatable on $m$ modes}. The requirement that $U_{SE}$ is a Gaussian unitary here is crucial; in fact, by Stinespring's dilation theorem~\cite{stinespring} every quantum channel can be represented as in Eq.~\eqref{G dilatable Stinespring} for some unitary $U_{AE}$ and some ancillary state $\sigma_E$.

\begin{figure}
\begin{center}
\scalebox{0.6}{\includegraphics[width=0.618\columnwidth]{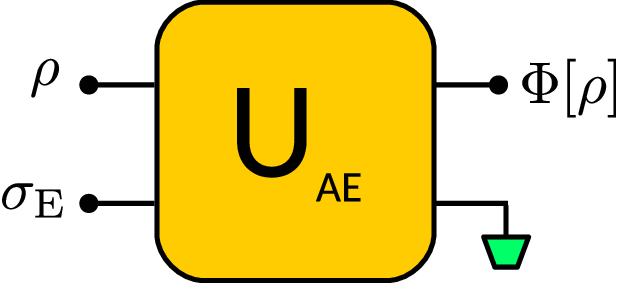}}
\end{center}
\caption{A Gaussian dilatable channel $\Phi$ is one that can be realized through a system-environment Gaussian unitary evolution $U_{AE}$ generated by a quadratic Hamiltonian, with any environment state ${\rho}_{E}$, and the environment degrees of freedom being subsequently discarded. The resulting channel is Gaussian or non-Gaussian depending on whether the environment state $\hr{E}$ is Gaussian or not. }
\label{fig3}
\end{figure}

We now derive the explicit action of a Gaussian dilatable channel on characteristic functions. Remember that Gaussian unitaries can be always factorised as $U_{AE} = \left(D_A(s)\otimes D_E(t)\right)\widetilde{U}_{AE}$, where $\widetilde{U}_{AE}$ is a symplectic unitary whose corresponding symplectic matrix we denote by $S_{AE}$. This means that the transformation $\rho_{AE} \mapsto \widetilde{U}_{AE} \rho_{AE} \widetilde{U}_{AE}^\dag$ translates to $\chi_\rho(\zeta)\mapsto\chi_\rho(S_{AE}\zeta)$ at the phase space level.
Decomposing $S_{AE}$ according to the splitting $A\oplus E$ as
\bb
S_{AE} = \begin{pmatrix} X & Y \\ Z & W \end{pmatrix} ,
\label{symplectic AE}
\ee
where $X$ is $2n\times 2n$ and $Y$ is $2n \times 2m$, it is not difficult to verify that the channel in Eq.~\eqref{G dilatable Stinespring} acts on characteristic functions as~\cite{pang}
\bb
\chi_\rho(\xi) \longmapsto e^{is^\intercal \Omega \xi} \chi_\rho (X\xi)\, \chi_{\sigma}(Y\xi)\, ,
\label{G dilatable chi}
\ee
where 
\bb
X^{\intercal}\Omega X + Y^{\intercal}\Omega Y = \Omega\, .
\label{XY}
\ee
In particular, \emph{every Gaussian dilatable channel is linear bosonic.} As we shall see in what follows, the converse is not true (Corollary~\ref{Es not G dilatable cor}). The main result of the present paper is however that every linear bosonic channel can be approximated by Gaussian dilatable channels to any desired degree of accuracy, in a precise sense (Theorem~\ref{approx n modes thm}). See Figure~\ref{fig1} for a pictorial representation of all the different classes of channels discussed in this paper.

\begin{note}
The two matrices $\Omega$ on the l.h.s.\ of Eq.~\eqref{XY} are of sizes $2n$ and $2m$, respectively. Without further specifying it, in what follows we will always assume that all matrices $\Omega$ are of the correct size.
\end{note}

It is worth noting that Eq.~\eqref{XY} is the only condition to be obeyed in order for the transformation in Eq.~\eqref{G dilatable chi} to derive from a Gaussian dilatable channel. In other words, Eq.~\eqref{XY} implies that one can find matrices $Z$ and $W$ with the property that the matrix in Eq.~\eqref{symplectic AE} is symplectic. This is a consequence of the completion theorem for symplectic bases~\cite[Theorem~1.15]{GOSSON}.

As we mentioned before, Gaussian channels are always Gaussian dilatable, and the corresponding state $\sigma_E$ of Eq.~\eqref{G dilatable Stinespring} can always be chosen to be Gaussian~\cite{min-dilations, partha}. Another important subclass of Gaussian dilatable channels is that composed of \textbf{passive dilatable channels}~\cite{wolfact, idler}, i.e.\ those for which Eq.~\eqref{G dilatable Stinespring} holds with $U_{AE}$ satisfying $[U_{AE}, H_{AE}]=0$ for the free-field Hamiltonian (number operator) $H_{AE}=\frac12 r^\intercal r$.

\begin{figure}
\begin{center}
\includegraphics[width=0.618\columnwidth]{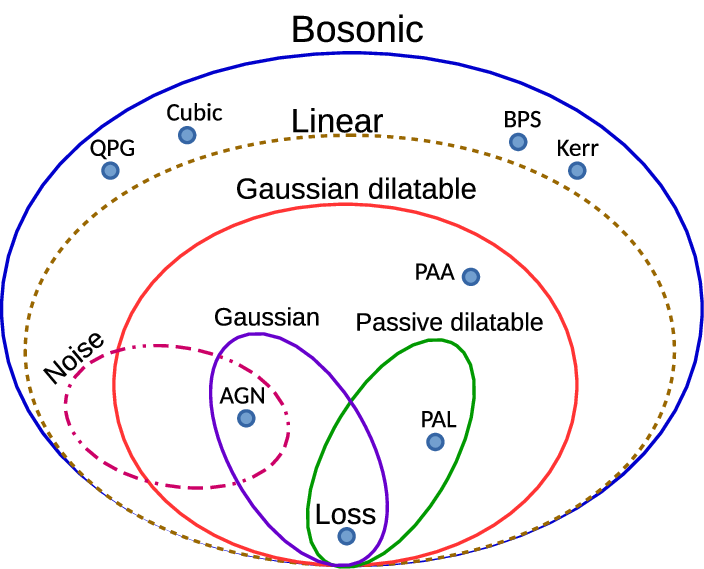}
\end{center}
\caption{The different classes of bosonic channels. Linear channels are those that implement linear transformations on the input signal at the phase space level. Gaussian dilatable channels are realised by Gaussian unitaries, resulting in a Gaussian channel if the environment is a Gaussian state. Restricting the Gaussian unitary to a passive unitary leads to a passive dilatable channel. PAL is the photon-added lossy channel and PAA is the photon-added amplifier channel. Cubic represents a cubic phase gate, QPG for quadrature phase gates (of which cubic phase gate is an example), and Kerr stands for the higher-order non-linear Kerr effect, BPS stands for the binary phase-shift channel, and all of these lie outside linear channels. The lossy channel is Gaussian dilatable and in particular passive dilatable, and Noise represents all additive noise channels of which AGN is the additive Gaussian noise channel. The dashed lines correspond to sets of channels whose relation to the other sets are previous unknown and explored in this article.}
\label{fig1}
\end{figure}

\subsection{Operator topologies in a nutshell} \label{subsec operator topologies}

As discussed above, the main result we present in this article is that every linear bosonic channel can be approximated to any desired degree of accuracy by Gaussian dilatable channels. To make this statement mathematically rigorous, we need to clarify in what sense this approximation holds. This corresponds to introducing a topology, i.e.\ a notion of convergence of (generalised) sequences, on the set of quantum channels. Since channels are (super)operators, the relevant concept is that of operator topology. 

In the present context, the two most important operator topologies to be given to the space $\mathcal{T}(\mathcal{H})$ of trace-class operators on some Hilbert space $\mathcal{H}$ are the \textbf{norm topology} (aka strong topology) and the \textbf{weak operator topology} (WOT). A sequence of trace-class operators $T_k$ is said to converge to $T\in\mathcal{T}(\mathcal{H})$ in the norm topology (or strongly) if $\|T-T_k\|_1 \tendsk{} 0$, where $\|T\|_1\coloneqq \Tr \sqrt{T^\dag T}$ is the trace norm. In this case we write also $T_k \tendsk{s} T$. The same sequence is instead said to converge to $T$ in the weak operator topology, and we write $T_k \tendsk{WOT} T$, if
\bb
\lim_{k\rightarrow\infty} \braket{\alpha| T_k |\beta} = \braket{\alpha|T|\beta} \qquad \forall\ \ket{\alpha},\, \ket{\beta}\in\mathcal{H}\, .
\label{WOT trace-class}
\ee
\begin{rem}It is not difficult to come up with sequences of density operators that do not converge strongly but tend to zero in the weak operator topology. An example is given by the sequence $T_k = \ketbra{k}{k}$ of projectors onto the Fock basis states of a single harmonic oscillator. In particular, neither the trace nor the trace norm are continuous with respect to the weak operator topology! For sequences of trace-class operators, it turns out that the above example captures the only way in which the above two notions of convergence can lead to different conclusions. 
\end{rem}

The following lemma is a well-known result from the theory of operators on Hilbert spaces. The acronym `SWOT' stands for `Strong/Weak Operator Topology'.

\begin{lemma}[SWOT convergence lemma] \label{SWOT convergence lemma}
Let $(\rho_k)_k\in \mathcal{T}(\mathcal{H}_n)$ be a sequence of density operators on the Hilbert space $\mathcal{H}_n$ associated with $n$ harmonic oscillators. Then the following are equivalent:
\begin{enumerate}[(i)]
\item $(\rho_k)_k$ converges to a density operator in the weak operator topology;
\item $(\rho_k)_k$ converges in norm to a trace-class operator;
\item the sequence $\left(\chi_{\rho_k}\right)_k$ of characteristic functions converges pointwise to a function that is continuous at $0$.
\end{enumerate}
If any of the above conditions is met, then the two limits in (i) and (ii) are the same, and their characteristic function coincides with the limit in (iii).
\end{lemma}

The equivalence between conditions (i) and (ii) seems to have been discovered by Davies~\cite[Lemma 4.3]{Davies1969}, and later generalised by Arazy~\cite{Arazy1981}. Cushen and Hudson~\cite[Theorem 2]{Cushen1971} proved that (ii) $\Leftrightarrow$ (iii) using L\'evy's continuity theorem~\cite[Theorem 19.1]{JACOD-PROTTER}. The full power of the implication (iii) $\Rightarrow$ (ii) has been exploited for instance in~\cite{Wolf2006}. In Appendix~\ref{app SWOT lemma} we provide a self-contained proof of Lemma~\ref{SWOT convergence lemma} for the sake of completeness.

In what follows we will often deal with convergence questions also at the level of superoperators acting on $\mathcal{T}(\mathcal{H})$ (e.g.\ quantum channels). The strongest notion of convergence in this context is that of \textbf{uniform convergence}. A sequence $(\Phi_k)_k$ of maps $\Phi_k:\mathcal{T}(\mathcal{H}_A)\rightarrow \mathcal{T}(\mathcal{H}_A)$ is said to converge to $\Phi:\mathcal{T}(\mathcal{H}_A)\rightarrow \mathcal{T}(\mathcal{H}_A)$ uniformly, and we write $\Phi_k \tendsk{u} \Phi$, if $\|\Phi_k - \Phi\|_\diamond \tendsk{} 0$, where the \textbf{diamond norm} is defined by~\cite{Aharonov1998, PAULSEN}
\bb
\left\| \Delta \right\|_\diamond \coloneqq \sup\left\{ \| (\Delta_A \otimes I_B)(\rho_{AB})\|_1:\ \text{$\mathcal{H}_B$ Hilbert space, $\rho_{AB}\in\mathcal{D}(\mathcal{H}_A\otimes \mathcal{H}_B)$} \right\} .
\label{diamond}
\ee
The name `uniform' is justified since the action of the map $\Phi$ is approximated by that of the maps $\Phi_k$ \emph{independently} of the input state and of the correlations it has with any external system. 

As detailed in~\cite{Shirokov2018, VV-E-diamond}, uniform convergence is often too stringent of a requirement to be useful for applications in quantum information theory. A weaker notion of convergence that is still operationally very relevant is the following: $(\Phi_k)_k$ is said to converge to $\Phi$ in the \textbf{strong operator topology}, and we write $\Phi_k \tendsk{SOT} \Phi$, if
\bb
\Phi_k(\rho)\tendsk{s} \Phi(\rho)\qquad \forall\ \rho\in\mathcal{D}(\mathcal{H})\, .
\label{SOT channels}
\ee
Observe that in the above definition we can replace $\mathcal{D}(\mathcal{H})$ with $\mathcal{T}(\mathcal{H})$, as every trace-class operator is a linear combination of at most two quantum states. Strong operator convergence implies that the action of the map $\Phi$ on any \emph{known} input state can be approximated by that of the maps $\Phi_k$; for this reason, it is a strictly weaker condition than uniform convergence. 
The strong operator topology is studied for instance in~\cite{Shirokov2008}. On the mathematical level, one of its main strengths is that the set of completely positive and trace-preserving maps is closed with respect to it, i.e.\ the strong limit of a sequence of quantum channels is another quantum channel. Another important feature is that compositions and tensor products behave well under limits. More precisely, for all pairs of sequences $(\Phi_k)_k, (\Psi_k)_k$ of trace-norm bounded maps $\Phi_k,\Psi_k:\mathcal{T}(\mathcal{H})\to \mathcal{T}(\mathcal{H})$ such that $\Phi_k \tendsk{SOT} \Phi$ and $\Psi_k \tendsk{SOT} \Psi$, one has~\cite{Shirokov2008,Shirokov2018,Wilde2018}
\begin{align}
    \Phi_k \circ \Psi_k &\tendsk{SOT} \Phi\circ \Psi\, , \label{SOT composition} \\
    \Phi_k \otimes \Psi_k &\tendsk{SOT} \Phi \otimes \Psi\, . \label{SOT tensor product}
\end{align}

Recently, the strong operator topology was also shown to be connected with the notion of uniform convergence on the set of quantum states of bounded energy~\cite{Shirokov2018, VV-E-diamond}, a result which greatly bolstered its physical interpretation. In our case, the energy can be defined with respect to the free-field Hamiltonian $H\coloneqq \frac12 r^\intercal r$, and uniform convergence on energy-bounded states is given by
\bb
\|\Phi_k - \Phi\|_{\diamond E} \tendsk{} 0\, ,
\label{uniform convergence energy-bounded}
\ee
where the energy-constrained diamond norm is defined via~\cite[Eq.~(1)]{Shirokov2018}. We remind the reader that for the special case of interest here, i.e.\ that of a free-field Hamiltonian, a related definition was previously put forward in~\cite{PLOB}. The resulting norm is equivalent to the one of~\cite{Shirokov2018, VV-E-diamond}, but because of some desirable properties of the latter, we have chosen that definition.
We briefly summarize the various notions of convergence for operators and superoperators in Table~\ref{table1}.

\begin{table}
\renewcommand{\arraystretch}{2}
\setlength\tabcolsep{1.5ex}
\begin{tabular}{ |c|c|c|c| }
\hline
Space & Topology/convergence &Notation& Definition \\ \hline
\multirow{2}{*}{Operators} & weak & $T_k \tendsk{WOT} T$ & $\braket{\alpha| T_k |\beta} \tendsk{} \braket{\alpha|T|\beta} \ \ \forall\, \ket{\alpha},\ket{\beta}\in\mathcal{H}\, $ \\
 & strong & $T_k \tendsk{s} T$ & $\|T-T_k\|_1 \tendsk{} 0$  \\ \hline
\multirow{3}{*}{Superoperators} & uniform& $\Phi_k \tendsk{u} \Phi$ & $\|\Phi_k - \Phi\|_\diamond \tendsk{} 0$ \\
 & strong & $\Phi_k \tendsk{SOT} \Phi$ & $\Phi_k(\rho)\tendsk{s} \Phi(\rho)\quad \forall\ \rho\in\mathcal{D}(\mathcal{H})\,$  \\
 & $\begin{array}{c} \text{uniform} \\[-3ex] \text{on energy-bounded states} \end{array}$ & $\Phi_k \tendsk{u,E} \Phi$ & $\|\Phi_k - \Phi\|_{\diamond E} \tendsk{} 0$ \\ 
\hline
\end{tabular}
\caption{Summary of the different notions of convergence/operator topologies on the space of trace-class operators and channels (superoperators) acting on a given Hilbert space $\mathcal{H}$. Here $\{T_k\}_k$ and $\{\Phi_k\}_k$ are a sequence of operators and superoperators respectively; $\|\cdot\|_1, \, \|\cdot\|_{\diamond}$, and $\|\cdot\|_{\diamond E}$ denote respectively the trace norm of operators, diamond norm and energy-constrained diamond norm of superoperators. The final column provides the how the sequences $\{T_k\}_k$ and $\{\Phi_k\}_k$ converge respectively to $T$ and $\Phi$. The details of the convergence are explained in this subsection.}
\label{table1}
\end{table}

In this paper we will be mostly interested in the convergence of sequences of bosonic channels to other bosonic channels. In this context a similar result to Lemma~\ref{SWOT convergence lemma} holds, as we now set out to establish. 

\begin{lemma}[SWOTTED convergence lemma] \label{SWOTTED convergence channels}
Let $(\Phi_k)_k$ be a sequence of bosonic channels $\Phi_k: \mathcal{T}(\mathcal{H}_n)\rightarrow \mathcal{T}(\mathcal{H}_n)$. Then the following are equivalent:
\begin{enumerate}[(i)]
\item $(\Phi_k)_k$ converges to a bosonic channel uniformly on energy-bounded states (i.e.\ Eq.~\eqref{uniform convergence energy-bounded} is satisfied);
\item $(\Phi_k)_k$ converges in the strong operator topology;
\item for all $\ket{\psi}\in\mathcal{H}_n$, the sequence of characteristic functions $\left(\chi_{\Phi_k(\ketbra{\psi}{\psi})}\right)_k$ converges pointwise to a function that is continuous at $0$.
\end{enumerate}
If any of the above conditions is met, then the limit in (i) is the same as that in (ii) (call it $\Phi$), and that in (iii) is $\chi_{\Phi(\ketbra{\psi}{\psi})}$.
\end{lemma}

\begin{note}
The acronym `SWOTTED' stands for `Strong / Weak Operator Topology / Topology of Energy-constrained Diamond norms'.
\end{note}

\begin{proof}
The fact that (i) and (ii) are equivalent is proved in~\cite[Proposition~3(B)]{Shirokov2018}. Observe that our Hamiltonian $H=\frac12 r^\intercal r$ satisfies the required conditions.
Moreover, applying Lemma~\ref{SWOT convergence lemma} with $\rho_k=\Phi_k(\ketbra{\psi}{\psi})$ we deduce that (ii) implies (iii). All that is left to prove is that (iii) implies (ii).

Start by observing that if (iii) holds then for all finite-rank operators $A$ the sequence $\chi_{\Phi_k(A)}$ converges pointwise to a function that is continuous at $0$. For a given density operator $\rho$ and some fixed $\epsilon>0$, choose a projector $\Pi$ onto a finite-dimensional space such that $\|\rho - \Pi \rho \Pi\|_1\leq \epsilon$ (this is possible because $\rho$ is trace-class). Since $\rho_k \coloneqq \Phi_k(\Pi \rho \Pi)$ satisfies Lemma~\ref{SWOT convergence lemma}(iii), from Lemma~\ref{SWOT convergence lemma}(ii) we deduce that $\Phi_k(\Pi \rho \Pi) \tendsk{s} \Phi(\Pi \rho \Pi)$, i.e.\ $\|\Phi_k(\Pi \rho \Pi) - \Phi(\Pi \rho \Pi)\|_1\leq \epsilon$ for all $k\geq k_0$. Using the fact that quantum channels are trace-norm contractions, we find
\begin{align*}
\|\Phi(\rho) - \Phi_k (\rho)\|_1 &\leq \|\Phi(\rho) - \Phi(\Pi\rho\Pi)\|_1 + \|\Phi(\Pi\rho\Pi) - \Phi_k (\Pi\rho\Pi)\|_1 + \|\Phi_k(\Pi\rho\Pi) - \Phi_k (\rho)\|_1 \\
&\leq 2\|\rho - \Pi\rho\Pi\|_1 + \|\Phi(\Pi\rho\Pi) - \Phi_k (\Pi\rho\Pi)\|_1 \\
&\leq 3\epsilon
\end{align*}
for all $k\geq k_0$. Since $\epsilon$ was arbitrary, we deduce that $\Phi_k(\rho)\tendsk{s}\Phi(\rho)$ for all density operators $\rho$, implying that $\Phi_k \tendsk{SOT}\Phi$.
\end{proof}

\begin{rem} \label{BK convergence rem}
Equipped with Lemma~\ref{SWOTTED convergence channels}, it is really elementary to show that the Braunstein-Kimble continuous variable teleportation protocol~\cite{BK-teleportation} performed using as a resource a two-mode squeezed state of energy $E$ implements a sequence of channels that converge to the identity in the strong operator topology~\cite{BK-teleportation,Wilde2018} or -- equivalently -- uniformly on bounded energy states, but not uniformly on all states~\cite{BK-teleportation,Pirandola2018}. Indeed, the transformation in~\cite[Eq.~(9)]{BK-teleportation} can be rewritten at the level of characteristic functions as
\bbb
\chi_{\rho}(\alpha) \mapsto \chi_\rho(\alpha) e^{-\bar{\sigma}|\alpha|^2} ,
\eee
where we temporarily reverted to the complex notation. Here, $\bar{\sigma}$ is a parameter that will be made to tend to $0$ to increase the accuracy of the approximation. Since the function on the r.h.s.\ clearly converges pointwise to that on the l.h.s.\ for all fixed input states $\rho$, Lemma~\ref{SWOTTED convergence channels} guarantees that the channel convergence happens with respect to the strong operator topology, or uniformly on energy-bounded states. These statements have been the subject of discussion in some recent papers~\cite{Pirandola2018, Wilde2018}.
\end{rem}

\begin{rem} \label{attenuator rem}
Another elementary consequence of Lemma~\ref{SWOTTED convergence channels} concerns the family of Gaussian channels known as quantum limited attenuators. These are channels acting on an arbitrary number of modes and defined by
\bb
    \mathcal{A}_\lambda : \chi_\rho(\xi)\longmapsto \chi_{\rho}\left( \sqrt\lambda \xi\right) e^{-\frac{1-\lambda}{4} \xi^\intercal\xi} ,
\label{attenuator}
\ee
where $0\leq \lambda\leq 1$. A swift application of Lemma~\ref{SWOTTED convergence channels} shows that
\bb
\mathcal{A}_\lambda\tends{SOT}{\lambda\to 1} \mathrm{id}\, ,
\label{strong convergence attenuator}
\ee
with $\mathrm{id}$ being the identity channel, a fact already mentioned in \cite[Section~IV.B]{VV-E-diamond}. In \cite[Proposition~1]{VV-E-diamond} it is also shown that the convergence in Eq.~\eqref{strong convergence attenuator} is not uniform.
\end{rem}


\section{Main results} \label{sec main}

We have seen that linear bosonic channels constitute a mathematically natural generalisation of the set of Gaussian channels. On the other hand, Gaussian dilatable channels are a physically meaningful class of operations for which there exists an operationally feasible implementation that requires only moderate resources (namely, a single non-Gaussian state and a Gaussian unitary). The main result of the present paper is that these two sets of operations are deeply related with each other, namely the latter is the strong operator closure of the former (Theorem~\ref{approx n modes thm}).

\subsection{The set of linear bosonic channels is closed} \label{subsec linear channels closed}

Implicit in the above statement is the claim that linear bosonic channels form themselves a closed set. Let us start by investigating this question, whose affirmative answer is the content of our first result. Our proof technique rests crucially upon the SWOTTED convergence lemma (Lemma~\ref{SWOTTED convergence channels}).

\begin{thm} \label{closed thm}
The set of all linear bosonic channels acting on an $n$-mode system is closed with respect to the strong operator topology.
\end{thm}

\begin{proof}
Consider a sequence of linear bosonic channels $(\Phi_{X_k, f_k})_k$, and assume that $\Phi_{X_k,f_k}\tendsk{SOT} \Phi$ for some completely positive trace-preserving map $\Phi:\tclass{n}\rightarrow \tclass{n}$. We have to prove that also $\Phi$ is a linear bosonic channel. To this end, start by considering the sequence of matrices $(X_k)_k$. If the sequence $\left((X_k^\intercal X_k)_{jj}\right)_{k}$ is unbounded for some fixed $j=1,\ldots, 2n$, then for almost all $\xi$ the sequence $\xi^\intercal X_k^\intercal X_k \xi$ is also unbounded (this does not happen only on the zero-measure hyperplane $\xi_j=0$). Remembering that the characteristic function of the vacuum state is given by Eq.~\eqref{chi vacuum}, we deduce that the sequence of complex numbers $\chi_{\ketbra{0}{0}}(X_k\xi)$ has a vanishing subsequence for almost all $\xi$. Since all functions $f_k$ are upper bounded by $1$ in modulus, the same is true for the sequence $\chi_{\ketbra{0}{0}}(X_k\xi) f_k(\xi)=\chi_{\Phi_{X_k,f_k}(\ketbra{0}{0})}(\xi)$. This however converges to $\chi_{\Phi(\ketbra{0}{0})}(\xi)$ by Lemma~\ref{SWOTTED convergence channels}, hence we conclude that $\chi_{\Phi(\ketbra{0}{0})}(\xi)=0$ for almost all $\xi$. As $\chi_{\Phi(\ketbra{0}{0})}$ is necessarily continuous, we conclude that it must vanish everywhere. This is absurd as $\chi_{\Phi(\ketbra{0}{0})}(0)=1$ by the quantum Bochner theorem (Lemma~\ref{quantum Bochner}).

We then conclude that $(X_k^\intercal X_k)_{jj}$ is bounded for all $j$, and hence there exists $M>0$ such that $\|X_k\|_{1\rightarrow 2} \coloneqq \max_j (X_k^\intercal X_k)_{jj}^{1/2} \leq M$. This means that the sequence $(X_k^\intercal X_k)_k$ is itself bounded, because $\|\cdot\|_{1\rightarrow 2}$ is a norm on the set of matrices. Since we are in a finite-dimensional space, there will exist a converging subsequence $X_{k_m}\tends{}{m\rightarrow\infty} X$.

We now show that $X$ is the limit of $(X_k)_k$. In fact, if this were not the case there would exist another limit point $X'\neq X$, i.e.\ $X_{k'_m}\tends{}{m\rightarrow\infty} X'\neq X$ for some other subsequence $(X_{k'_m})_m$. We can use the fact that $\Phi_{X_k, f_k}(\rho)\tendsk{s} \Phi(\rho)$ for \emph{all} density operators $\rho$ to deduce a contradiction. 
Choose as $\rho$ a Gaussian state with covariance matrix $\id$ and mean $s$, i.e.\ $\rho=\ketbra{s}{s}$ where $\ket{s}$ is a coherent state. Using the explicit expression in Eq.~\eqref{chi Gaussian} for the characteristic function together with Lemma~\ref{SWOTTED convergence channels}, we obtain
\bb
\lim_k e^{- \frac14 \xi^\intercal X_k^\intercal X_k \xi + is^\intercal \Omega X_k \xi} f_k(\xi) = \lim_k \chi_{\Phi_{X_k,f_k}(\ketbra{s}{s})}(\xi) = \chi_{\Phi(\ketbra{s}{s})}(\xi)\, .
\label{closed eq2}
\ee
Since on the two subsequences $(k_m)_m$ and $(k'_m)_m$ the exponential factor on the l.h.s.\ of the above equation converges to a nonzero limit, also the sequences $(f_{k_m}(\xi))_m$ and $(f_{k'_m}(\xi))_m$ have to converge. Let us define
\bb
g(\xi) \coloneqq \lim_m f_{k_m}(\xi)\, ,\qquad g'(\xi) \coloneqq \lim_m f_{k'_m}(\xi)\, .
\label{closed eq3}
\ee
Taking the limit in Eq.~\eqref{closed eq2} over these two subsequences yields
\begin{align*}
e^{- \frac14 \xi^\intercal X^\intercal X \xi + is^\intercal \Omega X \xi} g(\xi) &= \lim_m e^{- \frac14 \xi^\intercal X_{k_m}^\intercal X_{k_m} \xi + is^\intercal \Omega X_{k_m} \xi} f_{k_m}(\xi) \\
&= \chi_{\Phi(\ketbra{s}{s})}(\xi) \\
&= \lim_m e^{- \frac14 \xi^\intercal X_{k'_m}^\intercal X_{k'_m} \xi + is^\intercal \Omega X_{k'_m} \xi} f_{k'_m}(\xi) \\
&= e^{- \frac14 \xi^\intercal (X')^\intercal X' \xi + is^\intercal \Omega X' \xi} g'(\xi)\, .
\end{align*}
Since $\chi_{\Phi(\ketbra{s}{s})}(0)=1$ and $\chi_{\Phi(\ketbra{s}{s})}$ is continuous at $0$ by the quantum Bochner theorem (Lemma~\ref{quantum Bochner}), in an appropriate neighbourhood $\mathcal{U}$ of $0$ in $\mathbb{R}^{2n}$ it must hold that $\chi_{\Phi(\ketbra{s}{s})}(\xi)\neq 0$. From the above chain of equalities we deduce also that $g(\xi)\neq 0$ for all $\xi\in \mathcal{U}$. Hence
\bb
e^{\frac14 \xi^\intercal \left( (X')^\intercal X' -  X^\intercal X\right) \xi + is^\intercal \Omega \left(X - X'\right) \xi} = \frac{g'(\xi)}{g(\xi)}\qquad \forall\ \xi\in \mathcal{U}\, .
\label{closed eq4}
\ee
The problem with Eq.~\eqref{closed eq4} is that the r.h.s.\ does not depend explicitly on $s$. The only way in which this can happen on the l.h.s.\ as well is if $X=X'$, which is what we wanted to prove.

We have established that $X_k \tendsk{} X$. Then, the identity in Eq.~\eqref{closed eq2} can only hold if the sequence $(f_k)_k$ converges pointwise to some function $f$, i.e.\ $f_k(\xi) \tendsk{} f(\xi)$ for all fixed $\xi$. For an arbitrary state $\rho$ we then obtain
\bbb
\chi_{\Phi(\rho)}(\xi) = \lim_k \chi_{\Phi_{X_k, f_k}(\rho)}(\xi) = \lim_k \chi_\rho(X_k \xi) f_k(\xi) = \chi_\rho(X\xi) f(\xi)\, ,
\eee
where the last equality follows from the continuity of $\chi_\rho$. This shows that also $\Phi$ is a linear bosonic channel.
\end{proof}

\subsection{Approximate Gaussian dilatation for any linear bosonic channel} \label{subsec approximate}

Before we state our main result, we look at the simplified case where the matrix $J(X)$ of Eq.~\eqref{J(X)} is invertible. When this happens, it turns out that the problem of writing linear bosonic channels as (limits of) Gaussian dilatable channels simplifies considerably.

\begin{lemma} \emph{\cite{hw01}.} \label{J invertible}
Any linear bosonic channel $\Phi_{X,f}$ that acts on an $n$-mode system and satisfies $\det J(X)\neq 0$, where $J(X)$ is defined by Eq.~\eqref{J(X)}, is Gaussian dilatable using $n$ auxiliary  modes.
\end{lemma}

\begin{proof}
It is a well-known fact from elementary linear algebra that all invertible skew-symmetric matrices are equivalent up to congruences~\cite[Corollary~2.5.14(b)]{HJ1}. In particular, if $\det J(X)\neq 0$ there will exist a $2n\times 2n$ (invertible) matrix $Y$ such that
\bb
\Omega - X^\intercal \Omega X = J(X) = Y^\intercal \Omega Y\, .
\label{J and Y}
\ee
Remember that the complete positivity conditions for $\Phi_{X,f}$ as expressed by Lemma~\ref{lemma cp linear channels} imply that $f$ is $J(X)$-positive. It is not difficult to verify that the function $g:\mathbb{R}^{2n}\rightarrow \mathds{C}$ given by
\bb
g(\xi) \coloneqq f(Y^{-1}\xi)
\label{g and f}
\ee
will then be $\Omega$-positive. Since $g(0)=f(0)=1$ and moreover $g$ is clearly continuous at $0$ as the same is true for $f$, all conditions of the quantum Bochner theorem (Lemma~\ref{quantum Bochner}) are met, and hence $g(\xi)=\chi_\sigma(\xi)$ for some $n$-mode quantum state $\sigma$. Substituting $\xi\mapsto Y\xi$ we can then rewrite Eq.~\eqref{g and f} as $f(\xi)=\chi_\sigma(Y\xi)$. Inserting this into Eq.~\eqref{linear chi} we obtain Eq.~\eqref{G dilatable chi} for $s=0$. Since $X$ and $Y$ satisfy Eq.~\eqref{XY}, the channel $\Phi_{X,f}$ is Gaussian dilatable. This is summarized as algorithm~\ref{algo1}. 
\end{proof}

\begin{algorithm}
\label{algo1}
\KwIn{$X,f$ of a linear bosonic channel $\Phi_{X,f}$.}
~Compute the matrix $Y$ whose existence is guaranteed by Eq.~\eqref{J and Y}.\\
~Obtain $g$ using Eq.~\eqref{g and f}.\\
~Determine $\sigma$ (the ancillary state) from $g$ using Eq.~\eqref{integral}; this is possible by Lemma~\ref{quantum Bochner}.\\
~Complete $\left(\begin{smallmatrix} X & Y\\\cdot & \cdot \end{smallmatrix}\right)$ to a symplectic matrix $S_{AE}$.\\
~Obtain the Gaussian unitary $U_{AE}$ for the dilation from $S_{AE}$. 
\caption{Obtaining an exact Gaussian dilation for a linear bosonic channel with invertible $J(X)$.}
\KwOut{Gaussian dilation of $\Phi_{X,f}$.}
\end{algorithm}

\begin{rem} \label{direct cp conditions f rem1}
The above argument shows that any pair $X,f$ that satisfies conditions (i)-(iii) of Lemma~\ref{lemma cp linear channels} and such that $\det J(X)\neq 0$ induces a map $\Phi_{X,f}$ that admits the representation in Eq.~\eqref{G dilatable Stinespring} for some Gaussian unitary $U_{AE}$ and some quantum state $\sigma_E$. In particular, $\Phi_{X,f}$ must be a bosonic channel (completely positive and trace-preserving). This proves that conditions (i)-(iii) of Lemma~\ref{lemma cp linear channels} are sufficient in order for $\Phi_{X,f}$ to be a bosonic channel, at least when $\det J(X)\neq 0$. This observation can be used to give an independent proof of Lemma~\ref{lemma cp linear channels}, see Remark~\ref{direct cp conditions f rem2} and Appendix~\ref{app cp linear channels}.
\end{rem}

Lemma~\ref{J invertible} solves our problem as long as $J(X)$ is an invertible matrix. However, a quick inspection reveals that this is not always the case, the most notable counterexample being the identity channel, for which we have $X=\id$ and hence $J(X)=0$. We are now ready to state and prove our main result.

\begin{thm} \label{approx n modes thm}
The set of linear bosonic channels coincides with the closure of the set of Gaussian dilatable channels with respect to the strong operator topology.
\end{thm}

\begin{proof}
Since we know from Theorem~\ref{closed thm} that the set of linear bosonic channels is strong operator closed, we only have to show that the closure of the set of Gaussian dilatable channels contains it. 
To this end, pick a linear bosonic channel $\Phi_{X,f}$. 
For all sufficiently small $\epsilon>0$, let us consider the channels $\Phi_{X,f}\circ \mathcal{A}_{(1-\epsilon)^2}$, where $\mathcal{A}_\lambda$ is a quantum limited attenuator as defined by Eq.~\eqref{attenuator}, and $\circ$ denotes composition. It can be readily verified by concatenating the transformations in Eq.~\eqref{attenuator} and~\eqref{linear chi} that
\bb
\Phi_{X,f} \circ \mathcal{A}_{(1-\epsilon)^2} = \Phi_{X_\epsilon,\, f_\epsilon}
\ee
is another linear bosonic channel of parameters 
\begin{align}
    X_\epsilon &\coloneqq (1-\epsilon) X,\\
    f_\epsilon(\xi) &\coloneqq f(\xi) e^{-\frac14\left(1-(1-\epsilon)^2\right) \xi^\intercal X^\intercal X\xi} .
\end{align}
Even if the matrix $J(X)$ of Eq.~\eqref{J(X)} fails to be invertible, we claim that $J(X_\epsilon)$ is invertible for all sufficiently small $\epsilon>0$. This can be easily deduced by observing that $p(\epsilon)\coloneqq \det J(X_\epsilon)$ is a polynomial in $\epsilon$ which satisfies $p(1) = \det J(0) = \det \Omega =1$, hence it is not identically zero and thus has only isolated zeros.

Now, since $J(X_\epsilon)$ is invertible, Lemma~\ref{J invertible} implies that $\Phi_{X_\epsilon, f_\epsilon}$ is Gaussian dilatable. Moreover, since $\mathcal{A}_{(1-\epsilon)^2} \tends{SOT}{\epsilon \to 0} \mathrm{id}$ by Eq.~\eqref{strong convergence attenuator}, and channel composition behaves well under strong topology convergence, as formalised in Eq.~\eqref{SOT composition}, we have that
\bb
    \Phi_{X_\epsilon, f_\epsilon} = \Phi_{X,f}\circ \mathcal{A}_{(1-\epsilon)^2} \tends{SOT}{\epsilon\rightarrow 0} \Phi_{X,f} \circ \mathrm{id} = \Phi_{X,f}\, .
    \label{SOT convergence Phi}
\ee
This shows that any linear bosonic channel is the strong operator limit of a sequence of Gaussian dilatable channels, concluding the proof.
\end{proof}

\begin{rem} \label{direct cp conditions f rem2}
Continuing along the lines of Remark~\ref{direct cp conditions f rem1}, we argue that the above construct in fact proves also that conditions (i)-(iii) of Lemma~\ref{lemma cp linear channels} are sufficient in order for $\Phi_{X,f}$ to be a bosonic channel. This is because we constructed a sequence of strong operator approximations to $\Phi_{X,f}$ that are guaranteed to be bosonic channels (because they are Gaussian dilatable), and the set of all channels is closed with respect to the strong operator topology. For further details see Appendix~\ref{app cp linear channels}.
\end{rem}

Our proof of Theorem~\ref{approx n modes thm} is constructive, in that we give an explicit sequence of Gaussian dilatable channels that approximates a given linear bosonic channel. Among other things, this shows that $n$-mode ancillary states suffice to implement Gaussian dilatable approximations. We formalise this observation below and summarize the steps for obtaining the approximate dilation in Algorithm~\ref{algo1b}. 

\begin{corollary} \label{approx n modes cor}
Every linear bosonic channel on $n$ modes can be approximated to any desired degree of accuracy in the strong operator topology by channels that are Gaussian dilatable on $n$ modes only.
\end{corollary}

\begin{algorithm}
\label{algo1b}
\KwIn{$X,f$ of a linear bosonic channel $\Phi_{X,f}$.}
~Choose $\epsilon >0$ such that $J\left((1-\epsilon)X \right)$ is non-singular. \\
~Compose the original channel with an attenuator: $\Phi_{X_{\epsilon}, f_\epsilon}\coloneqq \Phi_{X,f}\circ \mathcal{A}_{(1-\epsilon)^2}$. \\
~Obtain a Gaussian dilation for $\Phi_{X_{\epsilon}, f_{\epsilon}}$ using Algorithm~\ref{algo1}. \\
~Limit $\epsilon \to 0^+$ provides the required approximation to $\Phi_{X,f}$ via Gaussian dilatable channels $\Phi_{X_{\epsilon}, f_{\epsilon}}$.
\caption{Obtaining an approximate Gaussian dilation for any linear bosonic channel using $n$ auxiliary modes.}
\KwOut{Approximate Gaussian dilation of $\Phi_{X,f}$.}
\end{algorithm}

\begin{rem}
Corollary~\ref{approx n modes cor} also applies to channels that are already Gaussian dilatable, but whose corresponding ancillary state (the $\sigma_E$ of Eq.~\eqref{G dilatable Stinespring}) is specified on a large number of modes $(>n)$ irrespective of it being Gaussian or not. The resulting resource compression could prove useful from a practical point of view. 
\end{rem}

\begin{rem} \label{approx dilations G channels rem}
A special case that is of great interest is that of Gaussian channels. In this context, Corollary~\ref{approx n modes cor} should be compared with the results of~\cite{min-dilations}. There the authors give an upper bound on the number of modes that are needed in order to construct a Gaussian dilation of any given Gaussian channel (with the ancillary state in Eq.~\eqref{G dilatable Stinespring} being also Gaussian). For instance, for a Gaussian additive noise channel, i.e.\ an additive noise channel as in Eq.~\eqref{additive noise chi} whose function $f$ is Gaussian, a $2n$-mode ancilla is required~\cite[Eq.~(48)]{min-dilations}. However, our findings show that it is possible to construct approximate Gaussian dilations of any $n$-mode Gaussian channel by means of an $n$-mode ancillary system only. Moreover, the proof of Theorem~\ref{approx n modes thm} shows that in this case the corresponding states $\sigma$ can be chosen to be Gaussian.
\end{rem}

\subsection{Why the closure is necessary} \label{subsec closure necessary}

A very natural question at this point is the following: \emph{is the closure in Theorem~\ref{approx n modes thm} really necessary?} In other words, \emph{are there examples of linear bosonic channels that are not Gaussian dilatable?} Here we argue that this is indeed the case by constructing an explicit example. This question has been recently investigated by the authors of~\cite{quntao18}, who developed a body of techniques to study non-Gaussian operations. Using these techniques, they were able to provide an example of a non-Gaussian channel (the `binary phase-shift channel' that applies a phase space inversion with probability $1/2$ on a single mode) that does not admit an exact Gaussian dilation in which the ancilla has finite energy. Observe that the channel they considered is not linear bosonic, so this result does not have immediate implications for the above question. Moreover, it does not seem possible to employ the techniques in~\cite{quntao18}, which rely crucially on the theory of relative entropy measures of non-Gaussianity~\cite{Genoni2008,Marian2013}, to exclude the existence of a Gaussian dilation whose corresponding ancillary state does not have finite first- or second-order moments.

One could argue that an exact Gaussian dilation such as that in Eq.~\eqref{G dilatable Stinespring} in which $\sigma$ has unbounded energy can not be physically realised if not approximately. However, in this case the approximation one can aim for is uniform instead of merely in the strong operator topology. Indeed, it is not difficult to realise that if $\Pi$ is a projector onto a finite-dimensional space such that $\sigma'\coloneqq \frac{\Pi \sigma \Pi}{\Tr[\sigma \Pi]}$ -- which has finite energy -- satisfies $\|\sigma-\sigma'\|_1\leq \epsilon$, the Gaussian channel $\Phi'$ defined by the same formula as in Eq.~\eqref{G dilatable Stinespring} with $\sigma'$ instead of $\sigma$ (and the same unitary $U$) satisfies $\|\Phi-\Phi'\|_\diamond\leq \epsilon$, where the diamond norm $\|\cdot\|_\diamond$ is defined by Eq.~\eqref{diamond}. In view of these considerations, we see that an exact Gaussian dilation for a quantum channel is still operationally meaningful even if the involved ancilla has infinite energy, because it leads to a sequence of uniform approximations via physically implementable Gaussian dilatable channels. This situation should be contrasted with the approximations we constructed in the proof of Theorem~\ref{approx n modes thm}, which are with respect to the strong operator topology and therefore `less accurate' in a precise sense.

In what follows we give an explicit example of a linear bosonic channel that is not exactly Gaussian dilatable, even if the ancillary state used for the dilation is allowed to have infinite energy. The channel we consider is a `binary displacement channel', which is a particular example of an additive noise channel (see Eq.~\eqref{additive noise p}) and is defined by the action
\bb
\mathcal{E}_s(\rho) \coloneqq \frac12 D(-s)\rho D(s) + \frac12 D(s)\rho D(-s)\qquad \forall\ \rho\in \mathcal{D}(\mathcal{H}_n)\, ,
\label{Es rho}
\ee
where $0\neq s\in\mathbb{R}^{2n}$ is a fixed vector. At the phase space level Eq.~\eqref{Es rho} translates to
\bb
\mathcal{E}_s:\chi_\rho(\xi)\longmapsto \chi_\rho(\xi) \cos\left( s^\intercal \Omega \xi\right) ,
\label{Es chi}
\ee
which proves that $\mathcal{E}_s$ is indeed a linear bosonic channel. We now set out to show that it can not be Gaussian dilatable.

\begin{lemma} \label{chi<1 lemma}
Any characteristic function $\chi_\sigma$ of an $m$-mode quantum state $\sigma$ satisfies
\bb
|\chi_\sigma(\zeta)|<1\qquad \forall\ \zeta\in \mathbb{R}^{2m}:\ \zeta\neq 0\, .
\label{chi<1}
\ee
\end{lemma}

\begin{proof}
We already commented on the fact that characteristic functions of quantum states are upper bounded by $1$ in modulus (see Lemma~\ref{continuity Omega positive lemma}). The problem is to show that there is strict inequality in Eq.~\eqref{chi<1} whenever $\zeta\neq 0$. Let $\sigma=\sum_{\mu=0}^\infty p_\mu \ketbra{\psi_\mu}{\psi_\mu}$ be the spectral decomposition of $\sigma$, where the series converges strongly (i.e.\ in trace norm). Let us write
\bbb
|\chi_\sigma(\zeta)| = \Tr [\sigma D(\zeta)] = \left| \sum_{\mu=0}^\infty p_\mu \braket{\psi_\mu | D(\zeta)| \psi_\mu}\right| \leq \sum_{\mu=0}^\infty p_\mu \left|\braket{\psi_\mu | D(\zeta)| \psi_\mu}\right| \leq \sum_{\mu=0}^\infty p_\mu = 1\, ,
\eee
where we used the fact that $D(\zeta)$ is unitary. From the above chain of inequalities it is clear that $|\chi_\sigma(\zeta)| = 1$ if and only if $D(\zeta)\ket{\psi_\mu}= e^{i\varphi_\mu} \ket{\psi_\mu}$ for all $\mu$ such that $p_\mu>0$, where $\varphi_\mu\in\mathbb{R}$. In particular, there exists a vector $\ket{\psi}\in\mathcal{H}_m$ such that $D(\zeta)\ket{\psi}= e^{i\varphi} \ket{\psi}$, implying that $D(\zeta) \ketbra{\psi}{\psi} D(-\zeta) = \ketbra{\psi}{\psi}$. Computing the trace against a displacement operator $D(\eta)$ of both sides of this identity and using Eq.~\eqref{Weyl} we obtain
\begin{align*}
e^{i\zeta^\intercal \Omega \eta} \chi_{\ketbra{\psi}{\psi}}(\eta) &= e^{i\zeta^\intercal \Omega \eta} \Tr \left[ D(\eta) \ketbra{\psi}{\psi}\right] \\
&= \Tr \left[ D(-\zeta) D(\eta) D(\zeta) \ketbra{\psi}{\psi}\right] \\
&= \Tr \left[ D(\eta) D(\zeta) \ketbra{\psi}{\psi} D(-\zeta) \right] \\
&= \chi_{D(\zeta) \ketbra{\psi}{\psi} D(-\zeta)}(\eta) \\
&= \chi_{\ketbra{\psi}{\psi}}(\eta)\, ,
\end{align*}
from which we deduce that $\chi_{\ketbra{\psi}{\psi}}(\eta)=0$ whenever $e^{i\zeta^\intercal \Omega \eta}\neq 1$, i.e.\ whenever $\left[D(\zeta), D(\eta)\right]\neq 0$. Since $\zeta\neq 0$ and $\Omega$ is invertible, this happens almost everywhere in $\eta$. The continuity of $\chi_{\ketbra{\psi}{\psi}}$ as established in Lemma~\ref{continuity Omega positive lemma} entails that $\chi_{\ketbra{\psi}{\psi}}(\zeta) \equiv 0$ for all $\zeta$, which is in contradiction with the fact that $\chi_{\ketbra{\psi}{\psi}}(0) =1$ as required by the quantum Bochner theorem (Lemma~\ref{quantum Bochner}).
\end{proof}

\begin{corollary} \label{Es not G dilatable cor}
The linear bosonic channel $\mathcal{E}_s$ defined in Eq.~\eqref{Es rho} is not Gaussian dilatable for $s\neq 0$.
\end{corollary}

\begin{proof}
Assume by contradiction that $\mathcal{E}_s$ were Gaussian dilatable. Comparing Eq.~\eqref{Es chi} with Eq.~\eqref{G dilatable chi}, we see that we would be able to find an $m$-mode state $\sigma$ and a $2m\times 2n$ matrix $Y$ such that $Y^\intercal \Omega Y=0$ and 
\bb
\chi_\sigma (Y\xi) = \cos\left( s^\intercal \Omega \xi\right) .
\label{Es not G dilatable eq1}
\ee
Observe that for all $\xi\in\mathbb{R}^{2n}$ such that $s^\intercal \Omega \xi\neq 0$ we have that $\tilde{\xi}\coloneqq \frac{2\pi}{s^\intercal \Omega \xi}\, \xi$ satisfies 
\bbb
\chi_\sigma (Y\tilde{\xi}) = \cos\left( s^\intercal \Omega \tilde{\xi}\right) = \cos(2\pi) =1\, .
\eee
Applying Lemma~\ref{chi<1 lemma}, we conclude that this is only possible if $Y\tilde{\xi}=0$ and hence $Y\xi=0$. Since the constraint $s^\intercal \Omega \xi\neq 0$ holds almost everywhere in $\xi$, and linear transformations are continuous, we conclude that $Y=0$ as a matrix. This is naturally in contradiction with Eq.~\eqref{Es not G dilatable eq1}, as the r.h.s.\ of the equation is not constant.
\end{proof}

Along the same lines, we suspect that the example in Eq.~\eqref{Es rho} can be generalised as to encompass all discrete convex combinations of conjugations by different displacement operators. We conjecture that all such linear bosonic channels are not exactly Gaussian dilatable.

\subsection{Number of auxiliary modes required} \label{subsec number auxiliary}

We now discuss the practical feasibility of approximating any linear bosonic channel with Gaussian dilatable maps. The main advantage of the technique we adopted to prove Theorem~\ref{approx n modes thm} is that it is fully explicit, and that it requires an ancillary system with the same number of modes as the system itself (Corollary~\ref{approx n modes cor}). However, the main disadvantage is that when $\det J(X)=0$ we had to construct a sequence of Gaussian dilations with varying Gaussian unitaries and varying ancillary states. This may be far from practical from an experimental point of view, as it would be more desirable to fix the unitary while varying only the ancillary state. It turns out that it is possible to do so, as we will now show. The following result is a generalisation of Lemma~\ref{J invertible}, and it is proved in the same spirit.

\begin{proposition} \label{approx n+k modes prop}
Let $\Phi_{X,f}$ be a linear bosonic channel acting on an $n$-mode system $A$. Let $E$ be an ancillary system composed of $n+k$ modes, where $k\coloneqq \frac12 \dim\ker J(X)$, and $J(X)$ is defined by Eq.~\eqref{J(X)}. Then there is a Gaussian unitary $U_{AE}$ and a sequence of states $\sigma_E(\epsilon)$ such that the corresponding Gaussian dilatable channels defined through Eq.~\eqref{G dilatable Stinespring} converge to $\Phi_{X,f}$ in the strong operator topology for $\epsilon\rightarrow 0^+$.
\end{proposition}

\begin{proof}
We start by looking at the $2n\times 2n$ skew-symmetric matrix $J(X)$ defined by Eq.~\eqref{J(X)}. We can find an orthogonal matrix $O\in\mathcal{O}(2n)$ such that~\cite[Corollary~2.5.14(b)]{HJ1}
\bb
J(X) = O^\intercal \left( \bigoplus\nolimits_{j=1}^n \begin{pmatrix} 0 & d_j \\ -d_j & 0 \end{pmatrix} \right) O\, ,
\label{jx-form}
\ee
where $d_j\geq 0$ are $n$ real numbers. Defining $k\coloneqq \frac12 \dim\ker J(X)$, we see that exactly $k$ of these numbers are zero. We now construct a $(2n+2k)\times 2n$ matrix $Y$ defined as follows:
\begin{align}
\label{yy'}
Y &\coloneqq \left( \bigoplus\nolimits_{j=1}^n Y'_j \right) O\, ,\\[1ex]
Y'_j &\coloneqq \left\{ \begin{array}{ll} d_j^{1/2} \id_2 & \text{if $d_j>0$,} \\[1ex] \left( \begin{smallmatrix} 1 & 0 \\ 0 & 0 \\ 0 & 1 \\ 0 & 0 \end{smallmatrix} \right) & \text{if $d_j=0$.} \end{array} \right. 
\end{align}
We first show that this $Y$ satisfies Eq.~\eqref{XY}. Before we do that, observe that the standard symplectic form of $n+k$ modes can be written as
\begin{align}
\Omega &\coloneqq \bigoplus\nolimits_{j=1}^n \Omega_j\, ,\\[1ex]
\Omega_j &\coloneqq \left\{ \begin{array}{ll} \left( \begin{smallmatrix} 0 & 1 \\ -1 & 0 \end{smallmatrix} \right) & \text{if $d_j>0$,} \\[2ex] \left( \begin{smallmatrix} 0 & 1 & & \\ -1 & 0 & & \\ & & 0 & 1 \\ & & -1 & 0 \end{smallmatrix}\right) & \text{if $d_j=0$.} \end{array}\right.
\end{align}
With this in mind, it is easy to check that indeed
\begin{align*}
Y^\intercal \Omega Y &= O^\intercal \left( \bigoplus\nolimits_{j=1}^n (Y'_j)^\intercal \Omega_j Y'_j  \right) O = O^\intercal \left( \bigoplus\nolimits_{j=1}^n \begin{pmatrix} 0 & d_j \\ -d_j & 0 \end{pmatrix} \right) O = \Omega - X^\intercal \Omega X\, .
\end{align*}
Now, consider the Moore--Penrose inverse of $Y$, call it $\widetilde{Y}$. One has
\begin{align}
\label{ytilde}
\widetilde{Y} &= O^\intercal \left( \bigoplus\nolimits_{j=1}^n \widetilde{Y}'_j \right) ,\\[1ex]
\widetilde{Y}'_j &= \left\{ \begin{array}{ll} d_j^{-1/2} \id_2 & \text{if $d_j>0$,} \\[1ex] \left( \begin{smallmatrix} 1 & 0 & 0 & 0 \\ 0 & 0 & 1 & 0 \end{smallmatrix} \right) & \text{if $d_j=0$.} \end{array} \right. 
\end{align}
Observe that
\begin{align}
\widetilde{Y} Y &= \id_{2n}\, , \label{Penrose 1}\\
Y \widetilde{Y} &= P\, , \label{Penrose 2}
\end{align}
where $\id_{2n}$ denotes the $2n\times 2n$ identity matrix, while $P$ is an orthogonal projector acting on $\mathbb{R}^{2n+2k}$. Define a family of $(2n+2k)\times (2n+2k)$ matrices $W(\epsilon)$ (where $\epsilon>0$ will later be taken to $0$) via the identities
\begin{align}
W(\epsilon) &\coloneqq \bigoplus\nolimits_{j=1}^n W_j(\epsilon)\, , \label{wepsilon 1} \\[1ex]
W_j(\epsilon) &\coloneqq \left\{\begin{array}{ll} \left( \begin{smallmatrix} 0 & \\ & 0 \end{smallmatrix}\right) & \text{if $d_j>0$,} \\[1ex] \left( \begin{smallmatrix} \epsilon & & & \\ & 1/\epsilon & & \\ & & \epsilon & \\ & & & 1/\epsilon \end{smallmatrix} \right) & \text{if $d_j=0$.} \end{array}\right. \label{wepsilon 2}
\end{align}
Because of the way it is constructed, it is elementary to verify that $W(\epsilon)$ obeys
\bb
W(\epsilon) \geq i \left( \Omega - P\Omega P\right) .
\label{Wg large}
\ee
Moreover, note that
\bb
Y^\intercal W(\epsilon) Y = \epsilon Q\, ,
\label{Wg sandwich}
\ee
where $Q$ is an orthogonal projector acting on $\mathbb{R}^{2n}$ and defined by
\begin{align}
Q &\coloneqq O^\intercal \left( \bigoplus\nolimits_{j=1}^n Q'_j\right) O\, ,\\[1ex]
Q'_j &\coloneqq \left\{ \begin{array}{ll} \left( \begin{smallmatrix} 0 & \\ & 0 \end{smallmatrix}\right) & \text{if $d_j>0$,} \\[1ex] \left( \begin{smallmatrix} 1 & \\ & 1 \end{smallmatrix} \right) & \text{if $d_j=0$.} \end{array} \right. 
\end{align}
We are in position to define the quantum state $\sigma$, which will pertain to an $(n+k)$-mode system. Let us set
\bb
\chi_{\sigma(\epsilon)}(\eta) \coloneqq f\big(\widetilde{Y}\eta\big) e^{-\frac14 \eta^\intercal W(\epsilon) \eta} .
\label{chiepsilon}
\ee
We now claim that for all $\epsilon>0$: (a) $\sigma(\epsilon)$ is a legitimate quantum state; and (b) $\chi_{\sigma(\epsilon)}(Y \xi) = f(\xi) e^{-\frac{\epsilon}{4} \xi^\intercal Q \xi}$ holds for all $\xi\in\mathbb{R}^{2n}$. To prove (a), it suffices to show that $\chi_{\sigma(\epsilon)}$ meets the conditions (i)-(iii) of the quantum Bochner theorem (Lemma~\ref{quantum Bochner}). Condition (i) is clear, since $\chi_{\sigma(\epsilon)}(0) = f(0) =1$; (ii) is also straightforward, since the continuity at $0$ of $f$ implies that of $\chi_{\sigma(\epsilon)}$. The problem is to verify (iii), i.e.\ $\Omega$-positivity. For a finite collection of vectors $\eta_1,\ldots, \eta_N\in\mathbb{R}^{2n+2k}$, let us write
\begin{align*}
\chi_{\sigma(\epsilon)}(\eta_\mu - \eta_\nu) e^{\frac{i}{2}\, \eta_\mu^\intercal \Omega \eta_\nu} &= f\left( \widetilde{Y}\eta_\mu - \widetilde{Y}\eta_\nu \right) e^{\frac{i}{2}\, \eta_\mu^\intercal \Omega \eta_\nu - \frac14 (\eta_\mu - \eta_\nu)^\intercal W(\epsilon) (\eta_\mu - \eta_\nu)} \\
&= f\left( \widetilde{Y}\eta_\mu - \widetilde{Y}\eta_\nu \right) e^{\frac{i}{2} \left( \widetilde{Y} \eta_\mu\right)^\intercal J(X) \left( \widetilde{Y} \eta_\nu\right)} \\
&\qquad \cdot e^{- \frac14 (\eta_\mu - \eta_\nu)^\intercal W(\epsilon) (\eta_\mu - \eta_\nu) - \frac{i}{2} \left( \widetilde{Y} \eta_\mu\right)^\intercal J(X) \left( \widetilde{Y} \eta_\nu\right) + \frac{i}{2}\, \eta_\mu^\intercal \Omega \eta_\nu}
\end{align*}
Since $f$ is $J(X)$-positive by Lemma~\ref{lemma cp linear channels}, we deduce that
\bbb
\left( f\left( \widetilde{Y}\eta_\mu - \widetilde{Y}\eta_\nu \right) e^{\frac{i}{2} \left( \widetilde{Y} \eta_\mu\right)^\intercal \left( \Omega - X^\intercal \Omega X\right) \left( \widetilde{Y} \eta_\nu\right)} \right)_{\mu,\nu} \geq 0\, .
\eee
To see why, it suffices to write out condition in Eq.~\eqref{A positivity} with $A=J(X)$ and $\xi_\mu = \widetilde{Y}\eta_\mu$. Using the fact that the Hadamard product of positive matrices is positive~\cite[Theorem~7.5.3]{HJ1}, we conclude that in order to show that $\chi_{\sigma(\epsilon)}$ is $\Omega$-positive it suffices to check that the matrix
\bbb
M \coloneqq \left( e^{- \frac14 (\eta_\mu - \eta_\nu)^\intercal W(\epsilon) (\eta_\mu - \eta_\nu) - \frac{i}{2}\, \left( \widetilde{Y} \eta_\mu\right)^\intercal J(X) \left( \widetilde{Y} \eta_\nu\right) + \frac{i}{2}\, \eta_\mu^\intercal \Omega \eta_\nu } \right)_{\mu,\nu}
\eee
is positive. Observe that
\begin{align*}
\widetilde{Y}^\intercal J(X) \widetilde{Y} = \widetilde{Y}^\intercal Y^\intercal \Omega Y \widetilde{Y} = \big( Y\widetilde{Y}\big)^\intercal \Omega \big(Y\widetilde{Y}\big) = P\Omega P\, .
\end{align*}
This allows us to rewrite
\begin{align*}
M_{\mu\nu} &= e^{- \frac14 (\eta_\mu - \eta_\nu)^\intercal W(\epsilon) (\eta_\mu - \eta_\nu) + \frac{i}{2}\, \eta_\mu^\intercal (\Omega - P\Omega P) \eta_\nu } .
\end{align*}
Since $W(\epsilon)\geq i(\Omega - P\Omega P)$ by Eq.~\eqref{Wg large}, it is not difficult to check (Lemma~\ref{A positivity lemma}) that $M\geq 0$, indeed. This concludes the proof of claim (a). The argument for claim (b) is much easier:
\begin{align*}
\chi_{\sigma(\epsilon)}(Y \xi) = f\big(\widetilde{Y}Y\xi\big) e^{-\frac{1}{4} \xi^\intercal Y^\intercal W(\epsilon) Y \xi} = f(\xi) e^{-\frac{\epsilon}{4} \xi^\intercal Q \xi} ,
\end{align*}
where in the second step we used Eq.~\eqref{Penrose 1} and Eq.~\eqref{Wg sandwich}.

Until now we have shown that if $\Phi_{X,f}$ is a linear bosonic channel then the channels acting as 
\bb
\chi_\rho(\xi) \longmapsto \chi_\rho( X \xi) f(\xi) e^{-\frac{\epsilon}{4} \xi^\intercal Q \xi}
\label{epsilon channel chi}
\ee
are Gaussian dilatable on $n+k$ modes for all $\epsilon>0$. We now complete the argument by proving that these approximate $\Phi_{X,f}$ in the strong operator topology. This is an immediate consequence of Lemma~\ref{SWOTTED convergence channels}(iii), since as $\epsilon\to 0^+$ the r.h.s.\ of Eq.~\eqref{epsilon channel chi} converges pointwise to $\chi_\rho( X \xi) f(\xi)$, which is manifestly continuous at $0$. For convenience we summarize in Algorithm~\ref{algo2} the steps to obtain the approximate dilation of a linear bosonic where the unitary in the dilation is fixed. 
 \end{proof}

\begin{algorithm}
\label{algo2}
\KwIn{$X,f$ of a linear bosonic channel $\Phi_{X,f}$.}
~Compute the canonical form for $J(X)$ from  Eq.~\eqref{jx-form}; set $k = \frac{1}{2} \text{dim\,ker} J(X)$. \\
~Obtain $\widetilde{Y}$ using Eq.~\eqref{ytilde}.\\
~For any choice of $\epsilon>0$ define $W(\epsilon)$ as in Eq.~\eqref{wepsilon 1}-\eqref{wepsilon 2}.\\
~Obtain the characteristic function of the ancillary $(n+k)$-mode state $\sigma(\epsilon)$ using Eq.~\eqref{chiepsilon}. \\
~Complete $\left(\begin{smallmatrix} X&Y\\\cdot & \cdot \end{smallmatrix}\right)$ to a symplectic matrix $S_{AE}$.\\
~Obtain the Gaussian unitary $U_{AE}$ for the dilation from $S_{AE}$. \\
~Limit $\epsilon \to 0^+$ provides the required approximation to $\Phi_{X,f}$ via Gaussian dilatable channels.
\caption{Obtaining an approximate Gaussian dilation for any linear bosonic channel using a fixed unitary in the dilation.}
\KwOut{Approximate Gaussian dilation of $\Phi_{X,f}$ with fixed unitary.}
\end{algorithm}

\section{Discussion and conclusions} \label{sec conclusions}

In this paper we investigated the set of linear bosonic channels by relating it to the physically meaningful set of Gaussian dilatable channels.
Our main result (Theorem~\ref{approx n modes thm}) states that the former set coincides with the strong operator closure of the latter. Operationally, this means that the action of every linear bosonic channel on any fixed state can be very well approximated by that of a sequence of suitable Gaussian dilatable channels (that do not depend on the state). We showed that taking the closure is in general necessary, as there are examples of linear bosonic channels for which no exact Gaussian dilation can be constructed, even if infinite-energy ancillary states are available (Corollary~\ref{Es not G dilatable cor}). Our results solve the open question raised in~\cite{pang} (see also~\cite[Remark~5]{Shirokov2013}) on the existence of Gaussian dilations for linear bosonic channels. Note that as it is formulated there,~\cite[Conjecture~1]{pang} is false, as Corollary~\ref{Es not G dilatable cor} shows that there are examples of linear bosonic channels that are not exactly Gaussian dilatable. However, the conjecture is `almost' true (i.e.\ it is true up to approximations).

Our proof strategy yields an explicit recipe to construct approximate Gaussian dilations of any given linear bosonic channel. We proved that if the associated Gaussian unitary is allowed to vary, an ancillary system with the same number of modes as the input system suffices (Corollary~\ref{approx n modes cor}). We can also require the unitary not to change when the approximation is sharpened, which yields a more experimentally friendly procedure. In this case we are still able to construct approximate Gaussian dilations, albeit with a larger number of ancillary modes (Proposition~\ref{approx n+k modes prop}). When applied to Gaussian channels, our findings complement those of~\cite{min-dilations}, in which only exact Gaussian dilations are considered (see Remark~\ref{approx dilations G channels rem}).


Some of the technical tools we developed are of broad interest in the field of quantum information with continuous variables. We highlight especially the SWOTTED convergence Lemma~\ref{SWOTTED convergence channels}, which gives easily verifiable necessary and sufficient conditions for a sequence of quantum channels to converge in the strong operator topology or -- equivalently -- uniformly on energy-bounded states. This is based on an analogous result for states (Lemma~\ref{SWOT convergence lemma}) that was known in the literature before. Our techniques also allowed us to give an alternative proof of the quantum Bochner theorem (Appendix~\ref{app quantum Bochner}). Our argument is entirely elementary and independent of the proof of the analogous result for classical probability theory, which requires -- one may argue -- more sophisticated measure theory tools. Our main theorem can also be used to prove directly that conditions (i)-(iii) of Lemma~\ref{lemma cp linear channels} suffice to ensure that the corresponding linear bosonic map is in fact a quantum channel. That they are also necessary is even easier to prove, as shown in Appendix~\ref{app cp linear channels}, thus we also obtain a proof of Lemma~\ref{lemma cp linear channels} that is independent of that presented in~\cite{Demoen77}. 

We now discuss some open problems. The careful reader may have noticed that in defining Gaussian dilatable channels (see Eq.~\eqref{G dilatable Stinespring}) we did not require the ancillary state to have finite energy. It would be interesting to investigate what happens when one has an energy constraint on the ancillary state. 
Another related aspect still to be explored is the most efficient way to implement linear bosonic channels via approximate Gaussian dilations, especially from an experimental point of view. In this context, we could for instance ask whether the constructions in the proofs of Theorem~\ref{approx n modes thm} and Proposition~\ref{approx n+k modes prop} are in some sense optimal, either from the point of view of the energy of the ancillary state, or from that of the number of modes required. Finally, our results could help to solve an interesting open question about the extremality of linear bosonic channels~\cite[Remark~2]{Holevo2013}.

In conclusion, this paper provides a number of novel insights into the structure of linear bosonic channels, further bolstering their status as an important subject of study in continuous variable quantum information.\\

{\it Acknowledgements.} We thank an anonymous referee for helpful comments, and in particular for suggesting a simplified proof of Theorem~\ref{approx n modes thm}. LL acknowledges financial support from the European Research Council (ERC) under the Starting Grant GQCOP (Grant No.~637352). AW was supported by the Spanish MINECO (project FIS2016-86681-P) with the support of FEDER funds, and the Generalitat de Catalunya (project 2017-SGR-1127).

\appendix
\addcontentsline{toc}{section}{Appendices}

\section{A direct proof of quantum Bochner theorem} \label{app quantum Bochner}

In this appendix we provide a self-contained proof of the quantum Bochner theorem (Lemma~\ref{quantum Bochner}) that is moreover entirely elementary, in that it requires only widely known results from standard analysis and no notion of measure theory. For comparison, the arguments in~\cite[Proposition~3.4(7)]{Werner84} and~\cite[\S 6.2.3]{deGosson2017} make use of the classical Bochner theorem (which in turn depends on Helly's selection principle or the Banach--Alaoglu theorem, see~\cite[Theorem~5.5.3]{BHATIA} for a standard proof), while that in~\cite[Theorem~5.4.1]{HOLEVO} employs the Stone--von Neumann uniqueness theorem. Remarkably, this is one of the few cases in which the quantum version of the statement is actually easier to prove that its classical counterpart.

\subsection{On some properties of $\Omega$-positive functions} \label{subsec properties Omega positive}

Before we delve into the proof of Lemma~\ref{quantum Bochner}, it is useful to familiarise with the important notion of $\Omega$-positive function. We start by studying the regularity properties of these functions, of which we made ample use in the main text. We need a preliminary lemma.

\begin{lemma} \label{lemma 3x3}
Let
\bb
A = \begin{pmatrix} a & x & y \\ x^{*} & a & z \\ y^{*} & z^{*} & a \end{pmatrix} \geq 0
\ee
be positive semidefinite. Then
\bb
|y-z|^{2}\leq \Re \left[(a-z)(a+z^{*}-2xy^{*}) \right] \leq 4 a |a-z|\, .
\ee
\end{lemma}

\begin{proof}
Follows by rearranging the inequality $\det A\geq 0$.
\end{proof}

\begin{lemma} \label{continuity Omega positive lemma}
Let $f:\mathbb{R}^{2n}\rightarrow \mathbb{C}$ be an $\Omega$-positive function. Then $|f(\xi)|\leq f(0)$ for all $\xi$, and moreover $f$ is continuous everywhere if and only if it is continuous at $0$.
\end{lemma}

\begin{proof}
For $A=\Omega$, $N=2$, $\xi_1=\xi$, and $\xi_2=0$, the condition in Eq.~\eqref{A positivity} reads
\bbb
\begin{pmatrix} f(0) & f(\xi) \\ f(\xi)^* & f(0) \end{pmatrix} \geq 0\, ,
\eee
from which we deduce that $f(0) \geq |f(\xi)|$. For $N=3$ and $\xi_{1}=0$, $\xi_{2}=-\zeta_{1}$ and $\xi_{3}=-\zeta_{2}$ we obtain instead
\bbb
\begin{pmatrix} f(0) & f(\zeta_1) & f(\zeta_2) \\ f(\zeta_1)^* & f(0) & f(\zeta_2-\zeta_1) e^{\frac{i}{2}\zeta_1^\intercal \Omega \zeta_2} \\ f(\zeta_2)^* & f(\zeta_2-\zeta_1)^* e^{-\frac{i}{2}\zeta_1^\intercal \Omega \zeta_2} & f(0) \end{pmatrix} \geq 0\, .
\eee
Applying Lemma~\ref{lemma 3x3} to the above matrix yields
\bbb
|f(\zeta_{1}) - f(\zeta_{2})|^{2} \leq 4 f(0) \left|f(0)- f(\zeta_{2}-\zeta_{1}) e^{\frac{i}{2} \zeta_{1}^{T}\Omega \zeta_{2}} \right|\, .
\eee
This shows that when $\zeta_k \rightarrow \zeta$ the sequence $f(\zeta_k)$ approaches $f(\zeta)$ at a rate that does not differ much from that of the convergence $f(\zeta-\zeta_k)\rightarrow 0$. Hence, $f$ is continuous everywhere whenever it is continuous at $0$.
\end{proof}

Before we move on to the proof of the quantum Bochner theorem, it is useful to make a simple sanity check. The characteristic function of a Gaussian state with zero displacement is given by $\chi_{\rho}(\xi)=e^{-\frac14 \xi^\intercal \Omega^\intercal V \Omega \xi}$, where $V\geq i\Omega$ is the covariance matrix. According to Lemma~\ref{quantum Bochner}, this must be an $\Omega$-positive definite function. Is there a way to verify this directly? The answer is affirmative, as we now set out to see.

\begin{lemma} \label{A positivity lemma}
Let $V$ and $A$ be $2n\times 2n$ real matrices. Assume $V$ is symmetric and $A$ skew-symmetric. The Gaussian function $f:\mathbb{R}^{2n}\rightarrow \mathbb{C}$ defined by $f(\xi)\coloneqq e^{-\frac14 \xi^\intercal V\xi}$ is $A$-positive in the sense of Definition~\ref{A positivity def} if and only if $V\geq i A$.
\end{lemma}

\begin{proof}
We first show that if $f$ is $A$-positive then necessarily $V\geq iA$. We evaluate the $A$-positivity condition of Eq.~\eqref{A positivity} on an arbitrary complex vector $x\in \mathbb{C}^N$ such that 
\bb
\sum_{\mu=1}^N x_\mu = 0\, ,
\label{x sums to 0}
\ee
obtaining
\bbb
\sum_{\mu,\nu} x_\mu^*x_\nu e^{-\frac{t^2}{4} (\xi_\mu-\xi_\nu)^\intercal V (\xi_\mu -\xi_\nu) + \frac{it^2}{2}\xi_\mu^\intercal A\xi_\nu} \geq 0
\eee
for all $t\in\mathbb{R}$ and finite collections $\xi_1,\ldots, \xi_N\in \mathbb{R}^{2n}$. Expanding to second order in $t$ around $0$ and using Eq.~\eqref{x sums to 0} yields
\bbb
\sum_{\mu,\nu} x_\mu^*x_\nu \left( -\frac{1}{2} (\xi_\mu-\xi_\nu)^\intercal V (\xi_\mu -\xi_\nu) + i \xi_\mu^\intercal A\xi_\nu\right) \geq 0\, .
\eee
Making use of Eq.~\eqref{x sums to 0} once again we can further simplify this to
\bbb
\sum_{\mu,\nu} x_\mu^*x_\nu \left( \xi_\mu^\intercal V \xi_\nu + i \xi_\mu^\intercal A\xi_\nu\right) = \left( \sum_\mu x_\mu \xi_\mu \right)^\dag (V+iA) \left( \sum_\nu x_\nu \xi_\nu\right) \geq 0\, .
\eee
Since every vector $z\in \mathbb{C}^{2n}$ can be written as $z=\sum_{\mu=1}^4 x_\mu \xi_\mu$ for $x = \frac12 (1,-1,i,-i)$ (which satisfies Eq.~\eqref{x sums to 0}), $\xi_1 = -\xi_2= \Re z$ and $\xi_3=-\xi_4=\Im z$, it must hold that $V+iA\geq 0$.

Conversely, let us prove that if $V+ iA\geq 0$ then the condition in Eq.~\eqref{A positivity} is always met. Start by rewriting
\bbb
f(\xi_\mu - \xi_\nu) e^{\frac{i}{2}\xi_\mu^\intercal A\xi_\nu} = e^{-\frac14 \xi_{\mu}^{\intercal} V \xi_{\mu}} e^{\frac12 \xi_{\mu}^{\intercal} (V + iA) \xi_{\nu}} e^{-\frac14 \xi_{\nu}^{\intercal} V \xi_{\nu}} \, .
\eee
Up to congruences by diagonal matrices, it is enough to show that $R_{\mu\nu} \coloneqq \exp\left[ \frac12 \xi_{\mu}^{\intercal} (V + iA) \xi_{\nu} \right]$ identifies a positive semidefinite matrix. Observe that $R$ is the Hadamard (i.e.\ entrywise) exponential of the Gram matrix $H_{\mu\nu}\coloneqq \frac12 \xi_{\mu}^{\intercal} (V+iA) \xi_{\nu}$ of a positive semidefinite form $V+iA$, hence it is itself positive semidefinite. Since Hadamard exponentials preserve positive semidefiniteness~\cite[Theorem~6.3.6]{HJ2}, we conclude that $R\geq 0$ as claimed.
\end{proof}

\subsection{Proof of quantum Bochner theorem} \label{subsec proof quantum Bochner}

We are ready to discuss our proof of Lemma~\ref{quantum Bochner}. Compared to other proofs that have appeared in the previous literature on the subject, ours relies heavily on the following elementary lemma, whose importance seems to have been somewhat overlooked.

\begin{lemma}[Diagonal integration lemma] \label{lemma single-double}
Let $f:\mathds{R}^{p}\rightarrow \mathds{C}$ be a bounded integrable function, i.e.\ let it be measurable and such that $|f(\xi)|\leq M$ and $\int d^{p}\xi\, |f(\xi)|<\infty$. Then
\bb
\int d^{p}\xi\, f(\xi) = \lim_{L\rightarrow \infty} \frac{1}{L^{p}} \int_{-L/2}^{L/2} d^{p}\xi_{1} d^{p}\xi_{2}\, f(\xi_{1}-\xi_{2})\, ,
\label{single-double}
\ee
where it is understood that the integration region on the right-hand side is $[-L/2, L/2]$ for all components of the vectors $\xi_{1}, \xi_{2}$.
\end{lemma}

\begin{proof}
We limit ourselves to showing the case $p=1$, since the others are analogous and follow by iteration of the same method. Let us write
\begin{align*}
\lim_{L\rightarrow\infty} \frac{1}{L} \int_{-L/2}^{L/2} d\xi_{1} \int_{-L/2}^{L/2} d\xi_{2}\, f(\xi_{1}-\xi_{2}) &\texteq{1} \lim_{L\rightarrow\infty} \frac{1}{L} \int_{-L}^{L} d\zeta \int_{-\frac{L+|\zeta|}{2}}^{\frac{L+|\zeta|}{2}} d\eta\, f(\zeta) \\
&= \lim_{L\rightarrow\infty} \frac{1}{L} \int_{-L}^{L} d\zeta\, (L+|\zeta|) f(\zeta) \\
&= \lim_{L\rightarrow\infty} \int_{-L}^{L} d\zeta \left(1+\frac{|\zeta|}{L}\right) f(\zeta) \\
&\texteq{2} \int_{-\infty}^{+\infty} d\zeta\, f(\zeta)\, .
\end{align*}
The justification of the above steps is as follows: 1: we defined $\zeta\coloneqq \xi_{1}-\xi_{2}$ and $\eta\coloneqq\frac{\xi_{1}+\xi_{2}}{2}$; 2: we observed that
\begin{align*}
\left| \int_{-\infty}^{+\infty}\!\!\! d\zeta\, f(\zeta) -\! \int_{-L}^{L}\!\! d\zeta \left(1+\frac{|\zeta|}{L}\right) f(\zeta) \right|&= \left| \int_{|\zeta|\geq L} d\zeta\, f(\zeta) - \int_{-L}^{L} d\zeta\, \frac{|\zeta|}{L} f(\zeta) \right| \\
&\leq \int_{|\zeta|\geq L} d\zeta\, |f(\zeta)| + \int_{-L}^{L} d\zeta\, \frac{|\zeta|}{L} |f(\zeta)| \\
&\leq \int_{|\zeta|\geq L}\!\! d\zeta\, |f(\zeta)| + \int_{L^{1/3} \leq |\zeta|\leq L}\!\! d\zeta\, |f(\zeta)| + \int_{-L^{1/3}}^{L^{1/3}}\!\! d\zeta\, \frac{|\zeta|}{L} |f(\zeta)| \\
&\leq \int_{|\zeta|\geq L^{1/3}} d\zeta\, |f(\zeta)| + \int_{-L^{1/3}}^{L^{1/3}}d\zeta\, \frac{|\zeta|}{L} \cdot M \\
&= \int_{|\zeta|\geq L^{1/3}} d\zeta\, |f(\zeta)| + \frac{M}{L^{1/3}}\, ,
\end{align*}
and both terms tend to zero as $L\rightarrow\infty$.
\end{proof}

\begin{proof}[Proof of Lemma~\ref{quantum Bochner}]
First of all, let us show that conditions (i)-(iii) are necessary for $f$ to be a characteristic function of a quantum state. This part of the proof is totally standard, see for instance~\cite[Theorem~5.4.1]{HOLEVO}. If $\rho$ is a normalised trace-class operator, by putting $\xi=0$ in Eq.~\eqref{chi} we find $\chi_{\rho}(0)=\Tr [\rho D(0)] = \Tr \rho = 1$, which proves the necessity of (i). Now, Stone's theorem ensures that the correspondence $\xi\mapsto D(\xi)$ is strongly operator continuous (i.e.\ $\lim_{\xi\rightarrow 0} \|\ket{\psi} - D(\xi)\ket{\psi}\|=0$ for all vectors $\ket{\psi}$, see \S~\ref{subsec operator topologies}), and in particular continuous in the weak operator topology. This latter fact can be equivalently rephrased as $\lim_{\xi\rightarrow 0} \Tr [ \rho D(\xi)] = \Tr [\rho]$ for all trace-class $\rho$ (see e.g.~\cite[Lemma~7]{approximate}), which is to say that $\chi_{\rho}$ is continuous at $0$ for all quantum states $\rho$. We move on to showing the necessity of requirement (iii). For $N\in\mathds{N}$, $\xi_{1},\ldots, \xi_{N}\in\mathds{R}^{2n}$ and $c\in\mathds{C}^{N}$, observe that
\begin{align*}
0 &\leq \Tr\left[ \left( \sum\nolimits_{\nu} c_{\nu} D(\xi_{\nu}) \right)^{\dag} \rho \left( \sum\nolimits_{\mu} c_{\mu} D(\xi_{\mu})\right) \right] \\
&= \sum_{\mu,\nu} c_{\mu} c_{\nu}^{*} \Tr [D(-\xi_{\nu})\,\rho\, D(\xi_{\mu})] \\
&= \sum_{\mu,\nu} c_{\mu} c_{\nu}^{*} \Tr [\rho\, D(\xi_{\mu})D(-\xi_{\nu})] \\
&= \sum_{\mu,\nu} c_{\mu} c_{\nu}^{*} e^{\frac{i}{2} \xi_{\mu}^\intercal\Omega \xi_{\nu}} \Tr [\rho\, D(\xi_{\mu}-\xi_{\nu})] \\
&= \sum_{\mu,\nu} c_{\mu} c_{\nu}^{*} e^{\frac{i}{2} \xi_{\mu}^\intercal\Omega \xi_{\nu}} \chi_{\rho}(\xi_{\mu}-\xi_{\nu})\, .
\end{align*}
This is equivalent to Eq.~\eqref{A positivity} for $A=\Omega$ and $f=\chi_{\rho}$.

Now we have to show the converse, i.e.\ that all functions $f$ satisfying (i)-(iii) are characteristic functions of some quantum state. By Lemma~\ref{continuity Omega positive lemma} we know that $f$ is continuous and bounded by $1$ in modulus. Let us first focus on the case when $f$ is square-integrable, i.e.\ such that $\int d^{2n}\xi\, |f(\xi)|^{2}<\infty$. This is more or less the strategy adopted in~\cite[Theorem~5.5.3]{BHATIA}. Under this assumption, one can construct
\bb
\rho_{f} \coloneqq \int \frac{d^{2n}\xi}{(2\pi)^{n}}\, f(\xi) D(-\xi)\, ,
\label{rhof}
\ee
which is well-defined in the sense of weak convergence (see~\cite[\S 5.3]{HOLEVO}, in particular Eq.~(5.3.18) and Theorem~5.3.3). Observe that $\rho_{f}$ is a Hilbert--Schmidt operator since $f$ is square-integrable, and that it satisfies
\bb
\Tr [\rho_{f} D(\xi)] = f(\xi)
\label{rhof Tr}
\ee
and
\bb
\Tr [\rho_{f}^{2}] =  \int \frac{d^{2n}\xi}{(2\pi)^{n}}\, |f(\xi)|^{2}\, ,
\label{rhof purity}
\ee
which is just a specialisation of Eq.~\eqref{Parseval}. We now show that $\rho_{f}$ is indeed positive semidefinite, i.e.\ that $\braket{\psi | \rho |\psi }\geq 0$ for all (normalised) vectors $\ket{\psi}$. Write
\begin{align*}
\braket{\psi | \rho_{f} |\psi } &\texteq{1} \int \frac{d^{2n}\xi}{(2\pi)^{n}} \, f(\xi)\, \braket{\psi|D(-\xi)|\psi} \\
&\texteq{2} \lim_{L\rightarrow \infty} \frac{1}{L^{2n}} \int_{-L/2}^{L/2} d^{2n}\xi_{1} d^{2n}\xi_{2}\, f(\xi_{1}-\xi_{2}) \braket{\psi|D(-\xi_{1}+\xi_{2})|\psi} \\
&\texteq{3} \lim_{L\rightarrow \infty} \frac{1}{L^{2n}} \int_{-L/2}^{L/2} d^{2n}\xi_{1} d^{2n}\xi_{2}\, f(\xi_{1}-\xi_{2}) e^{-\frac{i}{2} \xi_{1}^\intercal\Omega \xi_{2}} \braket{\psi|D(-\xi_{1})D(\xi_{2})|\psi} \\
&\texteq{4} \lim_{L\rightarrow \infty} \frac{1}{L^{2n}} \int_{-L/2}^{L/2} d^{2n}\xi_{1} d^{2n}\xi_{2}\, f(\xi_{1}-\xi_{2}) e^{-\frac{i}{2} \xi_{1}^\intercal\Omega \xi_{2}} \int \frac{d^{2n}\lambda}{(2\pi)^{n}}\braket{\psi |D(-\xi_{1})|\lambda}\braket{\lambda|D(\xi_{2})|\psi} \\
&\texteq{5} \lim_{L\rightarrow \infty} \frac{1}{L^{2n}} \int \frac{d^{2n}\lambda}{(2\pi)^{n}} \int_{-L/2}^{L/2} d^{2n}\xi_{1} d^{2n}\xi_{2}\, f(\xi_{1}-\xi_{2}) e^{-\frac{i}{2} \xi_{1}^\intercal\Omega \xi_{2}}\braket{\psi|D(-\xi_{1})|\lambda}\braket{\lambda|D(\xi_{2})|\psi} \\
&\textgeq{6} 0\, .
\end{align*}
The above steps are justified as follows: 1: is an application of Eq.~\eqref{rhof}; 2: follows from Lemma~\ref{lemma single-double} because $\xi\mapsto f(\xi) \braket{\psi|D(-\xi)|\psi}$ is integrable, as can be seen by writing
\begin{align*}
\int \frac{d^{2n}\xi}{(2\pi)^n} \left| f(\xi) \braket{\psi|D(-\xi)|\psi} \right| &\leq \left( \int \frac{d^{2n}\xi}{(2\pi)^n} \left| f(\xi) \right|^2 \right)^{1/2} \left( \int \frac{d^{2n}\xi}{(2\pi)^n} \left| \braket{\psi|D(-\xi)|\psi} \right|^2 \right)^{1/2} \\
&= \left( \int \frac{d^{2n}\xi}{(2\pi)^n} \left| f(\xi) \right|^2 \right)^{1/2} \\
&< \infty\, ,
\end{align*}
where we used Eq.~\eqref{Parseval} and the assumed square-integrability of $f$; 3: is an instance of Eq.~\eqref{Weyl}; 4: uses Eq.~\eqref{completeness explicit}; 5: is an application of Fubini--Tonelli theorem: since the iterated integral
\begin{align*}
&\int_{-L/2}^{L/2} d^{2n}\xi_{1} d^{2n}\xi_{2}\, \int \frac{d^{2n}\lambda}{(2\pi)^{n}} \left| f(\xi_{1}-\xi_{2}) e^{-\frac{i}{2} \xi_{1}^\intercal\Omega \xi_{2}} \braket{\psi |D(-\xi_{1})|\lambda}\braket{\lambda|D(\xi_{2})|\psi} \right| \\
&\qquad = \int_{-L/2}^{L/2} d^{2n}\xi_{1} d^{2n}\xi_{2}\,  |f(\xi_{1}-\xi_{2}) | \int \frac{d^{2n}\lambda}{(2\pi)^{n}} \left| \braket{\psi |D(-\xi_{1})|\lambda}\braket{\lambda|D(\xi_{2})|\psi} \right| \\
&\qquad \leq \int_{-L/2}^{L/2} d^{2n}\xi_{1} d^{2n}\xi_{2}\,  |f(\xi_{1}-\xi_{2}) | \left( \int \frac{d^{2n}\lambda}{(2\pi)^{n}} \left| \braket{\psi |D(-\xi_{1})|\lambda} \right|^{2}\right)^{1/2}  \left( \int \frac{d^{2n}\lambda}{(2\pi)^{n}} \left| \braket{\lambda |D(\xi_{2})|\psi} \right|^{2}\right)^{1/2} \\
&\qquad = \int_{-L/2}^{L/2} d^{2n}\xi_{1} d^{2n}\xi_{2}\,  |f(\xi_{1}-\xi_{2}) | \braket{\psi| D(-\xi_{1}) D(\xi_{1})|\psi}^{1/2} \braket{\psi| D(-\xi_{2}) D(\xi_{2})|\psi}^{1/2} \\
&\qquad =  \int_{-L/2}^{L/2} d^{2n}\xi_{1} d^{2n}\xi_{2}\,  |f(\xi_{1}-\xi_{2}) | \braket{\psi|\psi} \\
&\qquad \leq L^{4n}
\end{align*}
is clearly finite, we are free to exchange the integration order; finally, 6: is a consequence of the fact that 
\bbb
\int_{-L/2}^{L/2} d^{2n}\xi_{1} d^{2n}\xi_{2}\, f(\xi_{1}-\xi_{2}) e^{-\frac{i}{2} \xi_{1}^\intercal\Omega \xi_{2}}\braket{\psi|D(-\xi_{1})|\lambda}\braket{\lambda|D(\xi_{2})|\psi}\geq 0
\eee
for all $\lambda\in\mathds{R}^{2n}$, since this is an integral of a continuous bounded function over a bounded interval, hence it is a limit of Riemann sums, and each sum is positive because it is of the form $\sum_{\mu,\nu}c_{\mu}^{*}c_{\nu} f(\xi_{\mu}-\xi_{\nu}) e^{\frac{i}{2} \xi_{\mu}^\intercal\Omega\xi_{\nu}}\geq 0$, where the latter inequality holds because $f$ is assumed to be $\Omega$-positive. We have shown that $\rho_{f}$ defined in Eq.~\eqref{rhof} is positive semidefinite for all square-integrable functions $f$. In fact, in this case $\rho_{f}$ is actually trace-class, because $\Tr |\rho_{f}|=\Tr \rho_{f}=\Tr [\rho_{f} D(0)]=f(0)=1$, where we employed Eq.~\eqref{rhof Tr} together with hypothesis (i).

We now show that the working assumption that $f$ is square-integrable is actually a consequence of the properties of $f$. Since $|f|$ is bounded by $1$ by Lemma~\ref{continuity Omega positive lemma}, we can define the sequence of functions $f_{\epsilon}(\xi)\coloneqq f(\xi)\, e^{-\epsilon\, \xi^\intercal\xi}$, where $\epsilon>0$. For these functions, which are square-integrable and furthermore can be shown to be $\Omega$-positive as a consequence of Lemma~\ref{A positivity lemma}, one can construct a trace-class, positive semidefinite operator $\rho_{f_{\epsilon}}$ as in Eq.~\eqref{rhof}. As all legitimate density matrices, these will satisfy $\Tr [\rho_{f_{\epsilon}}^{2}]\leq 1$. Then, using Eq.~\eqref{rhof purity} one has
\bbb
1\geq \Tr \left[\rho_{f_{\epsilon}}^{2}\right] = \int \frac{d^{2n}\xi}{(2\pi)^{n}}\, |f_{\epsilon}(\xi)|^{2}\, .
\eee
For all $L>0$, using Lebesgue's dominated convergence theorem we can write
\bbb
\int^{L/2}_{L/2} \frac{d^{2n}\xi}{(2\pi)^{n}}\, |f(\xi)|^{2} = \lim_{\epsilon\rightarrow 0} \int^{L/2}_{L/2} \frac{d^{2n}\xi}{(2\pi)^{n}}\, |f_{\epsilon}(\xi)|^{2} \leq \lim_{\epsilon\rightarrow 0} \int \frac{d^{2n}\xi}{(2\pi)^{n}}\, |f_{\epsilon}(\xi)|^{2}  \leq 1\, .
\eee
Upon taking the limit $L\rightarrow\infty$, this shows that $f$ is indeed square-integrable, as claimed.
\end{proof}

\section{Complete positivity of linear bosonic maps} \label{app cp linear channels}

Throughout this appendix we show that an elementary and self-contained proof of Lemma~\ref{lemma cp linear channels} can be obtained as a by-product of our main result. Our argument should be compared with the original one by Demoen, Vanheuverzwijn and Verbeure~\cite{Demoen77} (see also~\cite{Evans1977, Demoen79}), which requires at the very least quite a few notions of functional analysis.

\begin{proof}[Proof of Lemma~\ref{lemma cp linear channels}]
As observed in Remark~\ref{direct cp conditions f rem2}, the proof of Theorem~\ref{approx n modes thm} combined with that of Lemma~\ref{J invertible} shows that conditions (i)-(iii) are sufficient for the map $\Phi_{X,f}$ to be a bosonic channel. Indeed, we showed that any such map is in the strong operator closure of the set of Gaussian dilatable channels, and quantum channels form a strong operator closed set.

Now we prove that conditions (i)-(iii) are also necessary. Since $\chi_{\Phi_{X,f}(\rho)}(\xi)=\chi_{\rho}(X\xi) f(\xi)$ must be the characteristic function of a quantum state for all input states $\rho$, (i) and (ii) follow in an elementary way from the quantum Bochner theorem (Lemma~\ref{quantum Bochner}). We then move on to (iii). First of all, we employ Eq.~\eqref{Weyl} and Eq.~\eqref{linear dag} to write 
\begin{align*}
e^{-\frac{i}{2} \xi_{\mu}^{\intercal}\Omega \xi_{\nu}} \Phi_{X,f}^{\dag}(D(\xi_{\mu}) D(-\xi_{\nu})) &= \Phi_{X,f}^{\dag}(D(\xi_{\mu}-\xi_{\nu})) \\
&= D(X\xi_{\mu}-X\xi_{\nu}) f(\xi_{\mu}-\xi_{\nu}) \\
&= D(X\xi_{\mu}) D(-X\xi_{\nu}) e^{-\frac{i}{2} \xi_{\mu}^{\intercal}X^{\intercal}\Omega X\xi_{\nu}} f(\xi_{\mu}-\xi_{\nu})\, ,
\end{align*}
which implies that
\bb
e^{\frac{i}{2} \xi_{\mu}^{\intercal} J(X) \xi_{\nu}} f(\xi_{\mu}-\xi_{\nu}) I = D(-X\xi_{\mu})\, \Phi_{X,f}^{\dag}(D(\xi_{\mu}) D(-\xi_{\nu}))\, D(X\xi_{\nu})\, ,
\label{useful identity}
\ee
where on the l.h.s\ $I$ stands for the identity operator on $\mathcal{H}_n$. If $\{\ket{0},\ket{1},\ldots\}$ is an orthonormal basis of $\mathcal{H}_n$ (called from now on the `computational basis'), define the truncated maximally entangled states on $\mathcal{H}_n\otimes \mathcal{H}_n$ as
\bb
\ket{\phi_{m}} \coloneqq \frac{1}{\sqrt{m}} \sum_{j=0}^{m-1} \ket{jj}\, .
\label{tme}
\ee
Observe that
\bb
\braket{\phi_{m} | A \otimes B | \phi_{m}} = \frac1m \Tr \left[\Pi_{m} A^{T} \Pi_{m} B\right]\, ,
\label{tme sandwich}
\ee
where
\bbb
\Pi_{m} \coloneqq \sum_{j=0}^{m-1} \ket{j}\!\!\bra{j}\, ,
\eee
and the transposition is taken with respect to the computational basis. Using Eq.~\eqref{tme sandwich} and the fact that $\Pi_m A\Phi_m\tends{s}{m\rightarrow\infty} A$ for all trace-class $A$, it is not difficult to see that
\bb
\lim_{m\rightarrow \infty} m\, \braket{\phi_{m} | A \otimes B | \phi_{m}} = \Tr [AB]\qquad \forall\ A\in\mathcal{T}(\mathcal{H}_n),\ \forall\ B\in\mathcal{B}(\mathcal{H}_n)\, ,
\label{tme and trace}
\ee
where $\mathcal{B}(\mathcal{H}_n)$ stands for the set of bounded operators.

Now, take an arbitrary test density matrix $\rho\in\mathcal{D}(\mathcal{H}_n)$, whose specific nature is irrelevant to us. For $\xi_{1},\ldots, \xi_{N}\in\mathbb{R}^{2n}$ and $c\in\mathbb{C}^N$, construct the bounded operator $B\coloneqq \sum_{\mu} c_{\mu} D(-X\xi_{\mu})^{*}\otimes D(-\xi_{\mu})$, where $*$ denotes complex conjugation again with respect to the computational basis. Write
\begin{align*}
0 &\textleq{1} \lim_{m\rightarrow\infty} m\, \Tr \left[ \left(I\otimes\Phi_{X,f}\right)\left(\ketbra{\phi_m}{\phi_m}\right)\, B^\dag (\rho^\intercal\otimes I) B \right] \\
&\texteq{2} \lim_{m\rightarrow\infty} m\, \bra{\phi_m} \left(I\otimes\Phi_{X,f}^\dag\right)\left(B^\dag (\rho^\intercal\otimes I) B\right)\ket{\phi_m} \\
&= \lim_{m\rightarrow\infty} m\, \sum_{\mu,\nu} c_\mu^*c_\nu \bra{\phi_m} D(-X\xi_\mu)^\intercal \rho^\intercal D(-X\xi_\nu) \otimes \Phi_{X,f}^\dag \left( D(\xi_\mu) D(-\xi_\nu)\right) \ket{\phi_m} \\
&= \sum_{\mu,\nu} c_\mu^*c_\nu \lim_{m\rightarrow\infty} m\, \bra{\phi_m} D(-X\xi_\mu)^\intercal \rho^\intercal D(-X\xi_\nu)^* \otimes \Phi_{X,f}^\dag \left( D(\xi_\mu) D(-\xi_\nu)\right) \ket{\phi_m} \\
&\texteq{3} \sum_{\mu,\nu} c_\mu^*c_\nu \Tr\left[ D(X\xi_\nu) \,\rho\, D(-X\xi_\mu)\, \Phi_{X,f}^\dag \left( D(\xi_\mu) D(-\xi_\nu)\right)\right] \\
&= \sum_{\mu,\nu} c_\mu^*c_\nu \Tr\left[ \rho\, D(-X\xi_\mu)\, \Phi_{X,f}^\dag \left( D(\xi_\mu) D(-\xi_\nu)\right)\, D(X\xi_\nu)\right] \\
&\texteq{4} \sum_{\mu,\nu} c_\mu^*c_\nu\, e^{\frac{i}{2} \xi_{\mu}^{\intercal} J(X) \xi_{\nu}} f(\xi_{\mu}-\xi_{\nu})\, .
\end{align*}
The justification of the above steps is as follows: 1: since $\Phi_{X,f}$ is assumed to be completely positive, the operator $\left(I\otimes \Phi_{X,f}\right)\left( \ketbra{\phi_m}{\phi_m}\right)$ must be positive semidefinite; 2: we applied the definition of adjoint of a linear map; 3: we used Eq.~\eqref{tme and trace}; and 4: we employed the identity in Eq.~\eqref{useful identity} together with the fact that $\Tr \rho = 1$. Since the above inequality holds for all $c\in\mathbb{C}^N$, we have that
\bbb
\left( f(\xi_\mu - \xi_\nu) e^{\frac{i}{2} \xi_\mu^\intercal J(X)\xi_\nu} \right)_{\mu,\nu}\geq 0\, ,
\eee
showing that $f$ is $J(X)$-positive.
\end{proof}

\section{Proof of the SWOT convergence lemma} \label{app SWOT lemma}

\begin{proof}[Proof of Lemma~\ref{SWOT convergence lemma}]

We start by showing that (i) $\Rightarrow$ (ii). Assume that $\rho_k \tendsk{WOT} \rho$, with $\rho$ being a density operator. Fix $\epsilon>0$, and pick a projector $\Pi$ onto a subspace of finite dimension such that $\left\|\rho - \Pi \rho\Pi\right\|_1<\epsilon$. We start by applying the triangle inequality:
\bbb
\|\rho - \rho_k\|_1 \leq \left\|\rho - \Pi \rho\Pi\right\|_1 + \left\| \Pi \rho\Pi -\Pi \rho_k \Pi \right\|_1 + \left\|\Pi\rho_k\Pi - \rho_k\right\|_1\, .
\eee
The first term of the sum on the r.h.s.\ is already small, while the second can also be made smaller than $\epsilon$ by taking $k\geq k_0$ for some large enough integer $k_0$ (this is because $\Pi \rho\Pi -\Pi \rho_k \Pi$ has finite rank). As for the third term, thinking of $\Pi \rho_k \Pi$ as a post-measurement state we see that it must be close to the initial state $\rho_k$ whenever the corresponding probability $\Tr \left[ \rho_k \Pi\right]$ is close to $1$. Since $\lim_k \Tr \left[ \rho_k \Pi\right]=\Tr\left[\rho\Pi\right]$ (as follows from convergence in the weak operator topology), we can require that $\Tr \left[ \rho_k \Pi\right]>1-\epsilon$ for $k\geq k_0$. Then the `gentle measurement lemma'~\cite[Lemma 9]{VV1999} yields
\bbb
\left\|\rho_k - \Pi\rho_k \Pi\right\|_1 \leq 2 \sqrt{1-\Tr [\rho_k \Pi]} < 2\sqrt{\epsilon}\, .
\eee
Putting all together shows that
\bbb
\|\rho - \rho_k\|_1 < 2\epsilon + 2\sqrt{\epsilon}\, ,
\eee
completing the proof that (i) $\Rightarrow$ (ii).

The implication (ii) $\Rightarrow$ (iii) is obvious, since if $\rho_k \tendsk{s} \rho $ then for all $\xi\in\mathbb{R}^{2n}$ one has that
\bbb
|\chi_{\rho_k}(\xi) - \chi_\rho(\xi)| = \Tr \left[ (\rho_k -\rho) D(\xi)\right] \leq \|\rho_k - \rho\|_1\tendsk{} 0\, .
\eee 

The only missing step is thus (iii) $\Rightarrow$ (i). To show this, assume that $\chi_{\rho_k}$ converges pointwise to a function $f$ that is continuous at $0$.
We start by showing that $f$ is the characteristic function of a quantum state. First, $f(0)=\lim_k \chi_{\rho_k}(0) = \lim_k 1 = 1$, so the normalisation condition is met. Secondly, $f$ is continuous at $0$ by hypothesis. Third, the numbers
\bbb
f(\xi_\mu-\xi_\nu) e^{\frac{i}{2} \xi_\mu^\intercal \Omega \xi_\nu} = \lim_k \chi_{\rho_k}(\xi_\mu-\xi_\nu) e^{\frac{i}{2} \xi_\mu^\intercal \Omega \xi_\nu}
\eee
are the entries of an $N\times N$ positive semidefinite matrix for all $\xi_1,\ldots, \xi_N$, because the set of positive semidefinite matrices of a fixed size is closed, and the numbers inside the limit on the r.h.s.\ of the above equation are the entries of a positive semidefinite matrix for all $k$. By the quantum Bochner theorem, we conclude that $f$ must be the characteristic function of a quantum state, i.e.\ $f =\chi_\rho$ for some density operator $\rho$.

We now show that $\rho_k \tendsk{WOT} \rho$. Pick a finite-rank operator $A$ and $\epsilon>0$, and for some radius $R>0$ yet to be determined write
\begin{align}
\left| \Tr \left[ (\rho-\rho_k) A \right] \right| &= \left| \int d\xi \left( \chi_\rho(\xi) - \chi_{\rho_k}(\xi)\right) \chi_A(-\xi) \right| \label{convergence thm eq1} \\
&\leq \int_{|\xi|\leq R} d\xi \left| \chi_\rho(\xi) - \chi_{\rho_k}(\xi)\right| |\chi_A(-\xi)| + \int_{|\xi|>R} d\xi \left| \chi_\rho(\xi) - \chi_{\rho_k}(\xi)\right| |\chi_A(-\xi)| \label{convergence thm eq2} \\
&\leq \int_{|\xi|\leq R} d\xi \left| \chi_\rho(\xi) - \chi_{\rho_k}(\xi)\right| |\chi_A(-\xi)| \\
&\qquad + \left(\int_{|\xi|>R} d\xi \left| \chi_\rho(\xi) - \chi_{\rho_k}(\xi)\right|^2\right)^{1/2} \left(\int_{|\xi|>R} d\xi |\chi_A(\xi)|^2 \right)^{1/2} \label{convergence thm eq3} \\
&\leq \int_{|\xi|\leq R} d\xi \left| \chi_\rho(\xi) - \chi_{\rho_k}(\xi)\right| |\chi_A(-\xi)| + 2 \left(\int_{|\xi|>R} d\xi |\chi_A(\xi)|^2 \right)^{1/2} , \label{convergence thm eq4}
\end{align}
where for the last step we observed that
\bbb
\left(\int_{|\xi|>R} d\xi \left| \chi_\rho(\xi) - \chi_{\rho_k}(\xi)\right|^2\right)^{1/2} \leq \left(\int d\xi \left| \chi_\rho(\xi) - \chi_{\rho_k}(\xi)\right|^2\right)^{1/2} = \|\rho-\rho_k\|_2\leq 2\, .
\eee
Now, since $\chi_A\in L^2(\mathbb{R}^{2n})$ is a square-integrable function, the second addend of Eq.~\eqref{convergence thm eq4} can be made smaller than $\epsilon/2$ by taking $R$ large enough. As for the first addend, we can apply Lebesgue's dominated convergence theorem and show that it converges to $0$ for all fixed $R$. Indeed, since the integrable functions $\left| \chi_\rho(\xi) - \chi_{\rho_k}(\xi)\right| |\chi_A(-\xi)|$ tend to $0$ pointwise, and moreover are bounded by $2$ on the whole ball $|\xi|\leq R$, we have
\bbb
\lim_k \int_{|\xi|\leq R} d\xi \left| \chi_\rho(\xi) - \chi_{\rho_k}(\xi)\right| |\chi_A(-\xi)| = 0\, .
\eee
and therefore $\int_{|\xi|\leq R} d\xi \left| \chi_\rho(\xi) - \chi_{\rho_k}(\xi)\right| |\chi_A(-\xi)|<\epsilon/2$ for $n\geq N$. Putting all together, we see that
\bbb
\left| \Tr \left[ (\rho-\rho_k) A \right] \right| < \epsilon
\eee
when $k\geq k_0$, which shows that
\bb
\Tr \left[\rho_k A\right] \tendsk{} \Tr \left[ \rho A\right]
\ee
for all finite-rank operators $A$. This is the same as saying that $\rho_k \tendsk{WOT} \rho$, thereby completing the proof.
\end{proof}

\bibliography{biblio}

\end{document}